\pdfoutput=1
\PassOptionsToPackage{hyphens}{url}
\documentclass[a4paper,UKenglish,cleveref,autoref]{lipics-v2021}

\hideLIPIcs
\nolinenumbers

\usepackage{booktabs}
\usepackage{colortbl}
\usepackage{longtable}
\usepackage{float}
\usepackage{microtype}
\usepackage{needspace}
\usepackage{tabularx}
\usepackage{xurl}
\makeatletter
\let\floatlib@headings\ps@headings
\def\ps@headings{\floatlib@headings
  \let\@oddfoot\@empty\let\@evenfoot\@empty}
\makeatother
\usepackage{tikz}
\usetikzlibrary{positioning,arrows.meta,fit,backgrounds,calc}

\lstdefinelanguage{Lean}{
  morekeywords={theorem,lemma,def,example,by,import,open,namespace,end,structure,inductive,where,fun,let,have,show,from,at,in,using,with,match,if,then,else,do,return,instance,class,abbrev,variable,section,noncomputable,private,protected,public,module,deriving,partial,termination_by,decreasing_by,set_option,calc},
  sensitive=true,
  morecomment=[l]{--},
  morecomment=[s]{/-}{-/},
  morestring=[b]",
  literate=
    {ℝ}{{$\mathbb{R}$}}1 {ℂ}{{$\mathbb{C}$}}1 {ℕ}{{$\mathbb{N}$}}1 {ℤ}{{$\mathbb{Z}$}}1 {ℚ}{{$\mathbb{Q}$}}1
    {∫}{{$\int$}}1 {∀}{{$\forall$}}1 {∃}{{$\exists$}}1 {∈}{{$\in$}}1 {∉}{{$\notin$}}1 {≤}{{$\leq$}}1 {≥}{{$\geq$}}1 {•}{{$\bullet$}}1 {≠}{{$\neq$}}1
    {↦}{{$\mapsto$}}1 {→}{{$\to$}}1 {←}{{$\leftarrow$}}1 {↔}{{$\leftrightarrow$}}1 {∧}{{$\wedge$}}1 {∨}{{$\vee$}}1 {¬}{{$\neg$}}1
    {⊢}{{$\vdash$}}1 {∑}{{$\sum$}}1 {∏}{{$\prod$}}1 {∂}{{$\partial$}}1 {∇}{{$\nabla$}}1 {Δ}{{$\Delta$}}1 {π}{{$\pi$}}1 {μ}{{$\mu$}}1
    {α}{{$\alpha$}}1 {β}{{$\beta$}}1 {γ}{{$\gamma$}}1 {δ}{{$\delta$}}1 {ε}{{$\varepsilon$}}1 {λ}{{$\lambda$}}1 {σ}{{$\sigma$}}1 {τ}{{$\tau$}}1 {ω}{{$\omega$}}1 {Ω}{{$\Omega$}}1 {φ}{{$\varphi$}}1 {ψ}{{$\psi$}}1 {θ}{{$\theta$}}1 {ρ}{{$\rho$}}1 {ξ}{{$\xi$}}1 {η}{{$\eta$}}1 {ζ}{{$\zeta$}}1 {κ}{{$\kappa$}}1 {ι}{{$\iota$}}1
    {⁻¹}{{$^{-1}$}}2 {²}{{$^{2}$}}1 {³}{{$^{3}$}}1 {₀}{{$_{0}$}}1 {₁}{{$_{1}$}}1 {₂}{{$_{2}$}}1 {ₙ}{{$_{n}$}}1 {ₖ}{{$_{k}$}}1 {ᵢ}{{$_{i}$}}1
    {≥}{{$\geq$}}1 {⬝}{{$\cdot$}}1 {≈}{{$\approx$}}1 {≡}{{$\equiv$}}1 {⁴}{{$^{4}$}}1 {⁵}{{$^{5}$}}1 {₃}{{$_{3}$}}1 {₄}{{$_{4}$}}1
    {⌊}{{$\lfloor$}}1 {⌋}{{$\rfloor$}}1 {⌈}{{$\lceil$}}1 {⌉}{{$\rceil$}}1 {∣}{{$\mid$}}1 {↑}{{$\uparrow$}}1 {⊤}{{$\top$}}1 {⊥}{{$\bot$}}1 {ℝ^n}{{$\mathbb{R}^n$}}2 {∂ₓ}{{$\partial_x$}}2 {∂ᵣ}{{$\partial_r$}}2 {ᵣ}{{$_{r}$}}1 {ₛ}{{$_{s}$}}1 {ₜ}{{$_{t}$}}1 {ₘ}{{$_{m}$}}1 {ⱼ}{{$_{j}$}}1 {Σ}{{$\Sigma$}}1 {Π}{{$\Pi$}}1 {Γ}{{$\Gamma$}}1 {Λ}{{$\Lambda$}}1 {ν}{{$\nu$}}1 {χ}{{$\chi$}}1 {υ}{{$\upsilon$}}1 {ϕ}{{$\phi$}}1 {ϵ}{{$\epsilon$}}1 {≪}{{$\ll$}}1 {≫}{{$\gg$}}1 {∼}{{$\sim$}}1 {–}{{--}}1 {—}{{---}}1 {’}{{'}}1 {“}{{``}}1 {”}{{''}}1 {…}{{$\ldots$}}1
    {⟨}{{$\langle$}}1 {⟩}{{$\rangle$}}1 {∘}{{$\circ$}}1 {·}{{$\cdot$}}1 {×}{{$\times$}}1 {√}{{$\surd$}}1 {∞}{{$\infty$}}1 {⋯}{{$\cdots$}}1 {ℓ}{{$\ell$}}1 {‖}{{$\|$}}1 {∅}{{$\emptyset$}}1 {⊆}{{$\subseteq$}}1 {∩}{{$\cap$}}1 {∪}{{$\cup$}}1 {⁺}{{$^{+}$}}1 {ᶜ}{{$^{c}$}}1 {∂ₜ}{{$\partial_t$}}2 {ᵀ}{{$^{T}$}}1 {▸}{{$\triangleright$}}1
}
\definecolor{leankw}{RGB}{90,40,140}
\definecolor{leancomment}{RGB}{90,110,90}
\definecolor{leanstring}{RGB}{140,60,20}
\definecolor{theoremblue}{RGB}{35,73,112}
\makeatletter
\thm@headfont{%
  \color{theoremblue}$\blacktriangleright$\nobreakspace\sffamily\bfseries}
\makeatother

\newcommand{\lean}[1]{\texttt{#1}}

\newcommand{\R}{\mathbb{R}}
\newcommand{\Q}{\mathbb{Q}}
\newcommand{\N}{\mathbb{N}}
\newcommand{\Z}{\mathbb{Z}}
\newcommand{\ulp}{\operatorname{ulp}}

\newcommand{\para}[1]{\smallskip\noindent\textit{#1}\hspace{0.6em}}

\title{FloatLib: Verified Floating-Point Arithmetic in Lean}
\titlerunning{FloatLib}

\author{Robert Joseph George\footnote{These authors contributed equally to this work.}}{California Institute of Technology, Pasadena, CA, USA}{}{}{}
\author{Will Adkisson\footnotemark[1]}{Washington University in St.\ Louis, St.\ Louis, MO, USA}{}{}{}
\author{Anima Anandkumar}{California Institute of Technology, Pasadena, CA, USA}{}{}{}

\authorrunning{R.\,J. George, W. Adkisson, and A. Anandkumar}
\Copyright{Robert Joseph George, Will Adkisson, and Anima Anandkumar}



\category{}
\relatedversion{}
\supplementdetails{Software}{https://github.com/lean-dojo/FloatLib}

\EventEditors{}
\EventNoEds{0}
\EventLongTitle{}
\EventShortTitle{}
\EventAcronym{}
\EventYear{2026}
\EventDate{}
\EventLocation{}
\EventLogo{}
\SeriesVolume{}
\ArticleNo{}

\begin{document}
\begingroup
\renewcommand{\LARGE}{\fontsize{16}{19}\selectfont}
\maketitle
\endgroup
\hypersetup{
  pdfauthor={Robert Joseph George, Will Adkisson, Anima Anandkumar},
  pdfsubject={FloatLib preprint}
}

\begin{abstract}
We present FloatLib, a verified arbitrary-precision floating-point
arithmetic library in Lean~4 that combines broad format coverage,
machine-checked correctness, and efficient certified execution. To our
knowledge, FloatLib is the first Lean library to unify IEEE binary and
decimal arithmetic, arbitrary-width posits, P3109, and user-defined formats
and rounding rules behind interchangeable certified software backends. Every
certified backend is proved equal to a complete encoded specification,
preserving signed zeros and exceptional values, while numerical theorems
connect execution to real rounding, error bounds, and exactness. FloatLib
combines exhaustive certified tables for small formats with verified word
and limb kernels based on guard-and-sticky invariants and independently
checked quotient candidates. Its posit development additionally proves
standard rounding thresholds and exact quire accumulation within capacity
for arbitrary widths. Across matched workloads, FloatLib achieves speedups
of up to $1.46\times$ over FLoPS and $116\times$ over Universal, while
remaining slower in some regimes such as binary arithmetic against MPFR.
Independent conformance testing includes more than 102 million TestFloat
evaluations with zero differences under the tested relation. We release the
library, proofs, benchmarks, evaluation data, and guide as open source.
\end{abstract}

\section{Introduction}\label{sec:intro}

Large language models now generate programs from problem
descriptions~\cite{li2022alphacode} and construct formal proofs in
Lean~\cite{ren2025deepseekproverv2}. For numerical programs, the statement
being proved must account for the arithmetic that executes. A proof of
a real-valued identity does not by itself describe an implementation
that loses small terms, overflows, or rounds at different points in the
computation. These effects cause failures in established software:
DeepStability documents numerical bugs in PyTorch and
TensorFlow~\cite{kloberdanz2022deepstability}. Verification therefore
needs specifications of the arithmetic choices and theorems relating
executed operations to their numerical meaning.

These choices are especially visible in mixed-precision machine learning.
Storing weights, activations, and gradients in fewer bits reduces memory
use and data movement, and can increase arithmetic throughput on hardware
that supports those formats~\cite{micikevicius2018mixedprecision}.
Small gradients can underflow to zero, however, and a nonzero update can
be lost when added to a much larger weight. Loss scaling, wider
accumulators, and higher-precision weight updates preserve information
through these computations~\cite{micikevicius2018mixedprecision}.
FP8 formats make different choices about exponent range, significand
precision, and exceptional values; microscaling adds a shared scale to
blocks of values~\cite{micikevicius2022fp8,ocp-mx-2023}.
The dot product shows why the accumulator also belongs in the
specification. Separate multiplications and additions round twice at
each step; a fused multiply-add rounds once~\cite{ieee754-2019}.
A posit quire can accumulate finite products exactly within its capacity
and round only the final sum~\cite{posit2022standard}.
These implementations share a real-valued formula but have different
rounding errors and exactness properties.

Numerical solvers likewise use several precisions within one algorithm.
Iterative refinement can factor a matrix at low precision and compute
residuals at higher precision, concentrating expensive arithmetic where
it improves the answer. Its error analysis identifies the precision and
conditioning assumptions under which the iterations recover an accurate
solution~\cite{carson2018threeprecisions}.
Rigorous interval computation requires each computed interval to enclose
the exact value. CoqInterval and ValidSDP use floating-point enclosures
and rounding-error bounds in machine-checked mathematical
proofs~\cite{martindorel2023enabling}.
Such analyses need arithmetic contracts that state their rounding,
range, and capacity assumptions explicitly.

Repeated scalar operations also make speed a design constraint:
allocation, conversion, and wide intermediate arithmetic recur inside
reductions and solver iterations. GMP addresses these costs through
algorithms chosen by operand size, machine-word arithmetic, and optimized
assembly kernels~\cite{gmp}. MPFR builds on GMP to provide efficient
arbitrary-precision arithmetic with correctly rounded
results~\cite{fousse2007mpfr}. Verified software must account for these
costs while preserving its numerical guarantees. Software implementations
also allow experiments with formats that have no dedicated hardware.
The cost matters within numerical proofs themselves: the CoqInterval and
ValidSDP experiments show how faster arithmetic can reduce whole-proof
checking time~\cite{martindorel2023enabling}.

To address these requirements, we developed \textbf{FloatLib, a verified
arbitrary-precision floating-point arithmetic library in Lean~4.}
Its scalar contracts cover quantization, mixed-precision operations, and
accumulation in verified machine-learning programs, with access to the
mathematical analysis available in mathlib~\cite{mathlib2020}.
Users can configure precision,
exponent widths, bias, and exceptional encodings, or define a
representation beyond a radix-and-exponent layout. The interfaces also
admit custom rounding rules and shared-scale contexts, covering small
quantized formats and wide reference formats within the same library.
Encoded specifications describe the computed values; refinement theorems
equate implementations with these specifications, and numerical theorems
connect the results to real arithmetic. Rounding bounds and exactness
results can then be used in loop invariants and accumulation arguments.
The custom-format affine proof applies the binary rounding bound, while
the quire theorem establishes exactness of the complete accumulation loop.

The mathematical starting point is Flocq's generic formats, rounding
theory, and verified binary arithmetic in
Coq~\cite{boldo2011flocq,boldo2017flocq}.
In Lean, FloatSpec develops Flocq-style theory and IEEE
models~\cite{floatspec2026}, while FLoPS formalizes P3109 semantics and
rounding, with executable arithmetic and refinement
proofs~\cite{chang2026flops,flops2026artifact}.
FloatLib combines extensible numerical representations with cost-based
selection among certified software implementations
(\Cref{tab:related-overview}). A rounding theorem can therefore remain
independent of the table, word kernel, or limb algorithm selected to
implement the operation.

Lean provides a compiled language as well as a proof
assistant~\cite{moura2021lean4}. Direct evaluation of an exact
specification, however, can perform unnecessary work. Exact integers and
rationals let a reference operation compute an exact intermediate and
round once; a large exponent gap may produce a long integer whose low
bits cannot affect the answer, a cost also noted in
Flocq~\cite{boldo2011flocq}. FloatLib retains the exact definitions for
reasoning and proves that bounded alignment, truncated intermediates,
and specialized storage produce the same encoded results.
Small domains admit certified tables; larger formats use word, limb,
or exact backends according to their applicability and cost.
The refinement proofs must preserve carries and ties while discarding
the low-bit computations that do not affect the result.

\textbf{We make four contributions:}
\begin{itemize}
\item \textbf{Broad format coverage and arbitrary precision.}
FloatLib is, to our knowledge, the first Lean library to combine
IEEE binary and decimal, arbitrary-width posits, P3109, and custom formats
and rounding rules with interchangeable certified software backends.
Shared representation and quantization interfaces also cover fixed point,
logarithmic, codebook, and block-scaled representations.
Custom binary layouts reuse rounding,
error, and exactness theorems (\Cref{sec:models}); distinct contracts
cover posit rounding thresholds, exact quire accumulation within
capacity, and other numerical families (\Cref{sec:formats}).
\item \textbf{Fast arithmetic with interchangeable backends.}
We certify tables for small domains and develop word and limb kernels
whose shared rounding invariants justify discarding low bits.
Quotient checks separate division correctness from the algorithm that
constructs a candidate. A planner selects among certified implementations
using their costs and applicability conditions
(\Cref{sec:execution}).
\item \textbf{Fully verified software backends.}
Every backend admitted by the certified interface has a Lean proof that
it returns the encoded specification's result on every input, including
exceptional cases. Operation-specific theorems establish real rounding,
error bounds, and exactness under their stated assumptions
(\Cref{sec:models,sec:execution,sec:formats}).
\item \textbf{Comparisons of speed and numerical correctness.}
We benchmark standard and custom binary and posit formats against MPFR,
Berkeley SoftFloat, and Universal under matching configurations.
Further comparisons cover P3109 arithmetic with FLoPS and scalar conversions with
TensorLib. Conformance checks include more than 102 million TestFloat
evaluations of the binary model; our posit comparisons identify rounding
and exponent-packing bugs in the tested SoftPosit generic routines
(\Cref{sec:evaluation}).
\end{itemize}

\begin{table}[!tb]
\centering
\small
\setlength{\tabcolsep}{4pt}
\renewcommand{\arraystretch}{1.08}
\caption{Format scope, execution, and proof contracts of related systems.
Custom formats include both new layouts and new representations or
rounding rules; \Cref{sec:related} explains the distinctions.}
\label{tab:related-overview}
\begin{tabularx}{\linewidth}{@{}
  >{\raggedright\arraybackslash}p{.18\linewidth}
  >{\raggedright\arraybackslash}p{.27\linewidth}
  >{\raggedright\arraybackslash}p{.23\linewidth}
  >{\raggedright\arraybackslash}X@{}}
\toprule
\textbf{Work} & \textbf{Formats and extensibility} & \textbf{Execution} &
  \textbf{Connection to proofs} \\
\midrule
\rowcolor{black!6}
\multicolumn{4}{@{}l}{\strut\textbf{Arithmetic libraries and interfaces}} \\
Flocq~\cite{boldo2011flocq,flocq422} &
  Generic radix and exponent function; arbitrary-precision binary &
  Rounding and binary arithmetic &
  Generic theory and verified operators \\
FloatSpec~\cite{floatspec2026} &
  Generic formats; IEEE binary &
  Exact operators; noncomputable rounded IEEE models &
  Real-value and Hoare-style specifications; native-model bridges \\
FLoPS~\cite{chang2026flops,flops2026artifact} &
  Parameterized P3109; rounding policies &
  Executable P3109 arithmetic &
  Refinement and stochastic error bounds \\
Primitive Floats~\cite{bertholon2019primfloat,martindorel2023enabling} &
  Binary64; interval computation &
  Native floating-point primitives &
  Primitive axioms connected to Flocq \\
TensorLib~\cite{tensorlib2026} &
  Low-precision tensor dtypes &
  Specialized casts; native arithmetic &
  Tensor interfaces over native scalars \\
\rowcolor{theoremblue!8}
\textbf{FloatLib (ours)} &
  Arbitrary precision; binary, decimal, posits, P3109;
  custom formats and rounding &
  Certified table, word, and limb kernels; cost-based selection &
  Encoded refinement, numerical bridges, custom contracts \\
\midrule
\rowcolor{black!6}
\multicolumn{4}{@{}l}{\strut\textbf{Reasoning about expressions and programs}} \\
Gappa~\cite{daumas2010gappa} &
  Exact and rounded expressions &
  Interval and error analysis &
  Generated proof certificates \\
VCFloat2~\cite{appel2024vcfloat2} &
  User-defined formats and operations &
  Reflective roundoff analysis &
  Error bounds for Coq and VST proofs \\
SymFPU~\cite{brain2019symfpu} &
  IEEE/SMT floating point &
  Concrete or symbolic bit vectors &
  Floating-point encodings for SMT reasoning \\
\bottomrule
\end{tabularx}
\end{table}
\section{Encoded arithmetic and real rounding}\label{sec:models}

\Cref{fig:architecture} shows the four layers: exact numerical
representations, word and limb kernels, format models, and certified
execution. Let $W_F$ denote the encoded values of a format $F$, and
$S_F^\circ$ its reference operation for $\circ$. Implementation theorems
identify kernels with $S_F^\circ$; numerical theorems describe the
mathematical interpretation of the result.

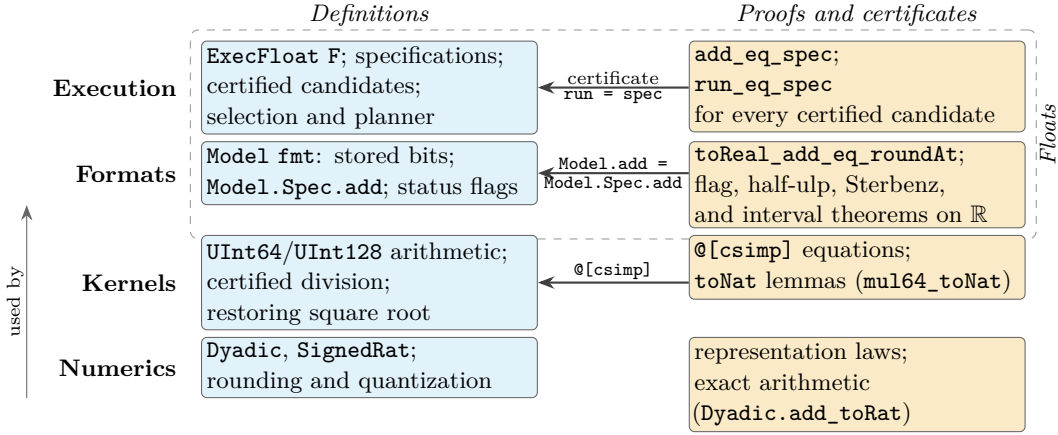
\begin{figure}[tbp]
\centering
\definecolor{ovRuntime}{RGB}{86,180,233}
\definecolor{ovProof}{RGB}{230,159,0}
\begin{tikzpicture}[
  font=\small,
  cell/.style={draw=black!60, rounded corners=2pt, align=left, inner sep=2pt,
               text width=4.3cm,
               execute at begin node={\hyphenpenalty=10000\exhyphenpenalty=10000}},
  rt/.style={cell, fill=ovRuntime!18},
  pf/.style={cell, fill=ovProof!22},
  row/.style={inner sep=0pt},
  layer/.style={font=\small\bfseries, anchor=east},
  head/.style={font=\small\itshape, align=center, text width=4.3cm},
  cert/.style={-{Stealth[length=2.2mm]}, thick, draw=black!75},
  dep/.style={-{Stealth[length=2mm]}, draw=black!60},
  lab/.style={font=\scriptsize, inner sep=1pt},
]
\node[head] (hrt) at (0,0) {Definitions};
\node[head] (hpf) at (6.45,0) {Proofs and certificates};

\node[rt, below=0.08cm of hrt] (e-rt) {%
  \lean{ExecFloat F}; specifications;\\
  certified candidates;\\
  selection and planner};
\node[pf, anchor=north] (e-pf) at (hpf |- e-rt.north) {%
  \lean{add\_eq\_spec};\\
  \lean{run\_eq\_spec}\\
  for every certified candidate};
\node[row, fit=(e-rt) (e-pf)] (row-e) {};

\node[rt, anchor=north] (f-rt) at ([yshift=-0.08cm]hrt |- row-e.south) {%
  \lean{Model fmt}: stored bits;\\
  \lean{Model.Spec.add}; status flags};
\node[pf, anchor=north] (f-pf) at (hpf |- f-rt.north) {%
  \lean{toReal\_\allowbreak add\_\allowbreak eq\_\allowbreak roundAt};\\
  flag, half-ulp, Sterbenz,\\
  and interval theorems on $\R$};
\node[row, fit=(f-rt) (f-pf)] (row-f) {};

\node[rt, anchor=north] (k-rt) at ([yshift=-0.08cm]hrt |- row-f.south) {%
  \lean{UInt64}/\lean{UInt128} arithmetic;\\
  certified division;\\
  restoring square root};
\node[pf, anchor=north] (k-pf) at (hpf |- k-rt.north) {%
  \lean{@[csimp]} equations;\\
  \lean{toNat} lemmas (\lean{mul64\_toNat})};
\node[row, fit=(k-rt) (k-pf)] (row-k) {};

\node[rt, anchor=north] (n-rt) at ([yshift=-0.08cm]hrt |- row-k.south) {%
  \lean{Dyadic}, \lean{SignedRat};\\
  rounding and quantization};
\node[pf, anchor=north] (n-pf) at (hpf |- n-rt.north) {%
  representation laws;\\
  exact arithmetic\\
  (\lean{Dyadic.add\_toRat})};

\node[layer, left=0.2cm of e-rt] (l-e) {Execution};
\node[layer, left=0.2cm of f-rt] (l-f) {Formats};
\node[layer, left=0.2cm of k-rt] (l-k) {Kernels};
\node[layer, left=0.2cm of n-rt] (l-n) {Numerics};

\draw[cert] (e-pf.west |- e-rt.east) -- node[lab, above] {certificate} node[lab, below] {\lean{run = spec}} (e-rt.east);
\draw[cert] (f-pf.west |- f-rt.east) -- node[lab, above] {\lean{Model.add =}} node[lab, below] {\lean{Model.Spec.add}} (f-rt.east);
\draw[cert] (k-pf.west |- k-rt.east) -- node[lab, above] {\lean{@[csimp]}} (k-rt.east);

\draw[dep] ([xshift=-0.3cm]l-n.west |- n-rt.south) -- node[lab, midway, rotate=90, anchor=south] {used by} ([xshift=-0.3cm]l-n.west |- f-rt.south);

\begin{scope}[on background layer]
  \node[draw=black!40, dashed, rounded corners=3pt, fit=(e-rt) (e-pf) (f-rt) (f-pf), inner sep=4pt] (floats) {};
\end{scope}
\node[rotate=90, anchor=north, font=\footnotesize\itshape, inner sep=1pt] at (floats.east) {Floats};
\end{tikzpicture}
\caption{FloatLib's four functional layers. Definitions (blue) appear
beside their proofs and certificates (yellow), including proof fields
erased at compilation. The upward arrow summarizes dependencies from
Numerics and Kernels into Floats; Formats and Execution also import one
another. Horizontal arrows show equations used to certify execution or
replace a definition during compilation.}
\label{fig:architecture}
\end{figure}

\Needspace{16\baselineskip}
\subsection{Representations and rounding contracts}\label{sec:representations}

A format's storage and meaning have separate interfaces in
\lean{FloatLib.Numerics} (namespace and setup declarations omitted):
\par\noindent\begin{minipage}{\linewidth}
\begin{lstlisting}
class EncodedFormat (F : Type u) where
  Code : Type v
  Scalar : Type w

class FormatSemantics (F : Type u) [EncodedFormat F] where
  denote : EncodedFormat.Code (F := F) →
    NumericalValue (EncodedFormat.Scalar (F := F))
\end{lstlisting}
\end{minipage}\par
\lean{NumericalValue} distinguishes finite scalars, infinities, and
exceptional values. Neither interface assumes a radix or bit layout.
A quantization contract relates an exact input to an output in a given
context; the context can carry a rounding mode, a shared scale, or random
bits. A new family supplies its denotation, operations, and proofs of
these contracts. Existing families already supply them for layouts
satisfying their parameter conditions.

A nonuniform four-entry codebook illustrates how a new representation
reuses the library's nearest-entry quantizer:
\par\noindent\begin{minipage}{\linewidth}
\begin{lstlisting}
import FloatLib
open FloatLib.Floats

def palette : Formats.Codebook 2 ℚ where
  denote w := .finite (match w.toNat with
    | 0 => -1 | 1 => 0 | 2 => 1 / 3 | _ => 2)
abbrev C := ExecFloat.Codebook palette
example : palette.nearestCode (1 / 2) = some 2 := by decide +kernel
\end{lstlisting}
\end{minipage}\par
The result is code~2, denoting $1/3$.
The generic \lean{nearestCode\_spec} theorem proves that the selected
finite entry minimizes distance to the input; ties select the lower
code. Users can replace that rule and prove a different quantization
contract over the same table. The carrier \lean{C} retains the table in
its type, so equal-width codebooks cannot be silently interchanged.

\subsection{Binary representations}

A binary descriptor records an exponent-field width $e_w$, a fraction-field
width $f_w$, a bias $b$, and an exceptional-encoding policy. Its stored
values are bit vectors of width $1+e_w+f_w$. The descriptor's proof fields
require $e_w\geq2$, $f_w>0$, and a bias that can encode the finite normal
number one. The same descriptor covers standard and custom layouts.
Binary32 instantiates $(e_w,f_w)=(8,23)$; the
custom 24-bit layout $(8,15)$ has precision~16 and uses the same
definitions and theorems.

Decoding produces a finite dyadic, an infinity, or a NaN. A finite dyadic
has the form $(-1)^s m2^q$, with natural significand $m$, integer exponent
$q$, and a Boolean sign $s$. Retaining the sign at $m=0$ distinguishes the
two zeros. The reference operations handle exceptional cases according
to the descriptor, compute finite intermediates exactly, and round and
pack the result. For fused multiply-add, the intermediate is the exact
dyadic $xy+z$, with one final rounding. Equality in $W_F$ preserves the
specification's choice of zero sign and exceptional encoding.
Status-returning operations additionally expose
five flags: invalid, division by zero, overflow, underflow, and inexact.
These are data returned with the result, not an implicit state of the
host processor. The binary status layer detects tininess after
rounding to the destination precision with an unbounded exponent;
underflow also requires inexactness.

The custom layout illustrates why the rounding point matters under
nearest-even rounding.
For $a=1+2^{-9}$, $x=1-2^{-9}$, and $b=-1$, all three operands are exact
and fused multiply-add returns $-2^{-18}$. Separate multiplication
rounds $1-2^{-18}$ to $1$, since the spacing immediately below one
is $2^{-16}$; adding $b$ then returns zero.
Both computations represent the real expression $ax+b$, but they
implement different operations.

The real interpretation, written $\operatorname{val}_F$, forgets the
sign of zero. Its Lean definition is total and assigns zero to nonfinite
values. Numerical theorems using this interpretation therefore state
finiteness explicitly. For directed
bounds that may overflow, the extended reals provide a more suitable
codomain.

\subsection{Generic rounding theory}

The mathematical theory follows Flocq's radix-and-exponent
formulation~\cite{boldo2011flocq,boldo2017flocq}. A radix $\beta\geq2$
and an exponent function $\varphi:\Z\to\Z$ define a representable grid.
For $x\neq0$, choose $e$ so that $\beta^{e-1}\leq |x|<\beta^e$
and set $c(x)=\varphi(e)$.
The value $x$ is representable when $x\beta^{-c(x)}$ is an integer.
Fixed point takes $\varphi(e)=e_{\min}$; an unbounded precision-$P$
format takes $\varphi(e)=e-P$; gradual underflow takes
\begin{equation}\label{eq:exponent-grid}
  \varphi(e)=\max(e-P,e_{\min}).
\end{equation}
Validity conditions on $\varphi$ make the local grids compatible across
magnitude intervals, with additional monotonicity hypotheses where
needed. Rounding reduces to an integer decision on the scaled significand:
\begin{equation}\label{eq:real-rounding}
  R(x)=\operatorname{rnd}\!\left(x\beta^{-c(x)}\right)\beta^{c(x)}.
\end{equation}
Nearest-even uses nearest-integer rounding with an even integer chosen
at a tie. Directed rounding replaces that integer operation by floor,
ceiling, or truncation. The generic theory proves representability, order
properties, and exactness on representable inputs. Its real-valued
statements use mathlib~\cite{mathlib2020}, so the rounding results can be
combined with its analysis library.

For a binary descriptor, set $P=f_w+1$, $e_{\min}=1-b-f_w$, and
$\beta=2$ in \eqref{eq:exponent-grid}. Write $R_F$ for nearest-even rounding
on this grid, with $R_F(0)=0$, and set
$\ulp_F(x)=2^{c(x)}$ for $x\neq0$ and
$\ulp_F(0)=2^{e_{\min}}$. The grid models gradual underflow but has no
upper exponent bound. It is therefore meaningful to round every real
number even when a corresponding finite word cannot be stored.

\begin{theorem}[Rounding error]\label{thm:rounding-error}
For every real $x$, $|R_F(x)-x|\leq \tfrac12\ulp_F(x)$.
If $x\neq0$ and $|x|\geq2^{1-b}$, then
$|R_F(x)-x|/|x|\leq2^{-P}$.
\end{theorem}
\begin{proof}
Nearest-integer rounding changes the scaled significand by at most
$1/2$; multiplying by the positive grid unit gives the absolute bound.
In the normal range, $c(x)=e-P$ and $|x|\geq2^{e-1}$ give the relative
bound. Below it, constant grid spacing preserves the absolute bound
but not a uniform relative bound.
\end{proof}

\subsection{Real interpretation of encoded operations}

The binary bridge applies to descriptors with IEEE exceptional encodings
and the conventional bias $2^{e_w-1}-1$. This is the library predicate
\lean{isIEEE}, which admits user-defined widths with these conventions
without requiring a standardized layout. The bridge theorem identifies
packing an exact dyadic with rounding its real value:
\begin{equation}\label{eq:dyadic-bridge}
  \operatorname{val}_F(\operatorname{roundDyadic}_F(d))
      =R_F(d_{\R}),
\end{equation}
provided the packed word is finite. The finiteness condition is essential:
the grid rounder has no overflow, whereas an encoded operation may return
infinity. The proof compares the integer quotient and remainder used by
the two rounders: normal packing, subnormal packing, and a carry into the
next exponent preserve the same nearest-even decision. A range argument
rules out exceptional packing under the finiteness premise.

\begin{theorem}[Finite arithmetic bridge]\label{thm:arithmetic-bridge}
Let $F$ satisfy the preceding IEEE conventions. If $x,y$ and their
reference sum are finite, then
\[
  \operatorname{val}_F(S_F^+(x,y))
    =R_F(\operatorname{val}_F(x)+\operatorname{val}_F(y)).
\]
The analogous equations hold for subtraction and multiplication when
their operands and respective results are finite.
For fused multiply-add with finite operands and result,
\[
  \operatorname{val}_F(S_F^{\mathrm{fma}}(x,y,z))
    =R_F(\operatorname{val}_F(x)\operatorname{val}_F(y)
          +\operatorname{val}_F(z)).
\]
\end{theorem}
\begin{proof}
The reference operation first computes its finite dyadic intermediate
exactly. Its real interpretation is the corresponding real sum,
difference, product, or product-and-sum. Applying
\eqref{eq:dyadic-bridge} supplies the single final rounding.
\end{proof}

Combining \Cref{thm:rounding-error,thm:arithmetic-bridge} bounds the error
by half an ulp of the exact result. In \Cref{sec:affine-proof}, a Lean example
applies these results to the custom affine operation without unfolding
packing or the selected backend.

Division has the same nearest-even connection when the divisor is
nonzero and the operands and result are finite. Square root requires a
finite operand that is nonnegative or a zero of either sign; its proof
also establishes result finiteness. Directed addition, subtraction,
multiplication, division, and square root instead provide outward
enclosures in the extended reals under their respective domain
conditions. Infinities can then be valid endpoints after overflow.
These statements concern the named arithmetic operations; elementary
function contracts are discussed in \Cref{sec:elementary}.

\subsection{Exact subtraction with gradual underflow}

Sterbenz's lemma states that subtracting nearby floating-point values is exact
\cite{sterbenz1974,muller2018handbook}.

\begin{theorem}[Encoded Sterbenz lemma]\label{thm:sterbenz}
Let $x,y$ be finite words in an IEEE-style descriptor, with positive
values $u=\operatorname{val}_F(x)$ and $v=\operatorname{val}_F(y)$.
If $u\leq2v$ and $v\leq2u$, then
$\operatorname{val}_F(S_F^-(x,y))=u-v$.
No separate finiteness assumption on the result is needed.
\end{theorem}
\begin{proof}
Both operands, and hence their difference, are integer multiples of
$2^{e_{\min}}$. If $|u-v|\leq2^{e_{\min}+P-1}$, this lattice property
places the difference in the gradual-underflow grid. For larger
differences, embed the operands in the unbounded precision-$P$ format,
apply its Sterbenz theorem, and use the magnitude bound to return to
the destination grid. Rounding therefore leaves the difference
unchanged. Since $|u-v|\leq\max(u,v)$ also bounds it by a finite input,
subtraction cannot overflow; the arithmetic bridge applies with this
derived finiteness proof.
\end{proof}

Cancellation can produce a subnormal even from normal operands, so
replacing gradual underflow by flush-to-zero would invalidate the
conclusion. The theorem concerns the real value of the difference;
the encoded operation also determines the sign of an exact zero.
\section{Efficient execution by refinement}\label{sec:execution}

Backend refinement proves that a cheaper implementation returns
the same encoded result. This equality lets the planner choose a kernel
while preserving the numerical theorems used to reason about the program.

\subsection{Execution certificates and backend selection}

The common interface, \path{ExecFloat.Backend.Certified}, pairs an
implementation with a proof of equality to its specification
(comments and setup declarations omitted):
\par\noindent\begin{minipage}{\linewidth}
\begin{lstlisting}
structure Certified {α : Type u} (spec : α) where
  estimate : Candidate
  run : α
  run_eq_spec : run = spec
\end{lstlisting}
\end{minipage}\par
The type $\alpha$ describes the whole operation, whether binary
arithmetic, a unary function, FMA, or an accumulator update.
The function equality \lean{run\_eq\_spec} covers every operand,
including exceptional values; \lean{Candidate} contains the cost
estimates used for selection.

A codec connects a family's chosen carrier, such as a machine word, to
its encoded model and proves the round-trip laws. Lifting a model
operation through this codec gives the specification on the carrier;
backend certificates prove equality with that lifted operation.

The planner folds over certified alternatives to a mandatory certified
baseline, comparing estimated warm and cold costs under a user-selected
policy and applying memory ceilings to replacements. Its result still
has type \lean{Certified spec}: the selection theorem is
\lean{run\_eq\_spec}, regardless of the cost model's accuracy.
The baseline remains available even if it exceeds the ceilings, so these
are selection preferences rather than proved resource bounds.

Kernel eligibility depends on descriptor fields and the operation. For
example, a binary two-word kernel requires the IEEE-style conventions,
more than 64 stored fraction bits, and at most 128 encoded bits. These
conditions permit a proof over all layouts that fit the representation.
An eligible kernel can still decline an operand outside its specialized
case; its complete dispatcher then calls the reference implementation.
The equality theorem covers both branches. Direct entry points and
inlining let closed configurations avoid constructing a planning record
at each call, while the same certificate identifies their meaning.

\subsection{Tables, words, and limbs}

A binary operation on an eight-bit domain has only $2^{16}$ input pairs:
a complete table replaces decoding, arithmetic, and rounding by indexing,
with each entry proved equal to the reference result. For larger layouts,
word kernels keep significands and products in bounded unsigned words
with explicit carries and borrows. At wider precisions, limb arrays
avoid converting every intermediate to the reference representation.

Aligning $a2^d$ and $b$ exactly can construct an
integer whose length grows with the exponent separation $d$.
If cancellation cannot expose the omitted bits, a short retained prefix
and one sticky bit suffice for the final rounding. A cancellation guard
justifies computing with this summary. FMA reuses the argument after
forming an exact product, preserving its single rounding.

Division instead separates proposal from acceptance: a bounded
multiplication and remainder comparison check a proposed quotient, and
a proved divider handles rejection (\Cref{fig:division}).
We measure the complete operations on a common scalar workload
(\Cref{sec:evaluation});
the timings do not isolate the contribution of each optimization.

\subsection{Guard and sticky invariants}
\label{sec:math-sticky}

A kernel often needs only a short prefix of an exact intermediate.
The problem is to summarize the discarded suffix without changing a
later rounding decision. For a natural number $x$ and a shift $j$, define
\[
 J_j(x)=
 \begin{cases}
   \lfloor x/2^j\rfloor,&x\bmod2^j=0,\\
   \lfloor x/2^j\rfloor\mathbin{\mathrm{OR}}1,&x\bmod2^j\neq0.
 \end{cases}
\]
This is \emph{shift right with jamming}: a nonzero discarded suffix sets
the low retained bit. Each word or limb implementation is proved equal
to this representation-independent definition.

To see what must be preserved, write, for $s>0$,
\[
 x=q2^s+g2^{s-1}+r,
 \qquad g\in\{0,1\},\quad 0\leq r<2^{s-1}.
\]
Rounding $x/2^s$ to nearest-even gives
\begin{equation}\label{eq:guard-sticky}
 E_s(x)=q+\mathbf{1}\bigl[g=1\ \land\ (r\neq0\ \lor\ q\text{ odd})\bigr].
\end{equation}
The retained quotient, guard bit $g$, and nonzeroness of the remaining
suffix determine the answer. \Cref{fig:sticky}(a) shows which information
survives the shorter representation.

\begin{theorem}[Rounding after jamming]
\label{thm:sticky}
For $x,j,s\in\N$ with $j+2\leq s$,
\mbox{$E_{s-j}(J_j(x))=E_s(x)$}.
\end{theorem}
\begin{proof}
For $k>0$, setting the low bit cannot affect the quotient above position
$k$ or the bit at position $k$. The shared integer lemmas establish
\[
 \begin{aligned}
 \left\lfloor J_j(x)/2^k\right\rfloor
     &=\left\lfloor x/2^{j+k}\right\rfloor,\\
 \operatorname{bit}(J_j(x),k)&=\operatorname{bit}(x,j+k),\\
 J_j(x)\bmod2^k\neq0
   &\quad\Longleftrightarrow\quad x\bmod2^{j+k}\neq0.
 \end{aligned}
\]
The last line is the suffix invariant. If the initial shift was exact,
it follows from divisibility.
Otherwise the jammed word is odd and the original discarded suffix is
nonzero. Apply these identities at the retained quotient and guard
positions, and substitute into \eqref{eq:guard-sticky}. The inequality
$j+2\leq s$ places the jam strictly below the guard.
\end{proof}

The separation condition cannot be weakened in general. For $x=5$, $j=1$, and $s=2$,
jamming gives $J_1(5)=3$. Rounding $3/2$ to even gives $2$, whereas
rounding $5/4$ gives $1$. Here the jam has replaced the guard bit.

The wide-limb subtraction kernel gives a concrete use of the theorem.
It stores 32-bit limbs and admits IEEE-style descriptors wider than
128 bits with exponent fields of at most 32 bits. Let $p$ be the stored
fraction width, $a,b\geq2^p$ the magnitudes to be subtracted, and
$d=s_a-s_b\geq0$ their scale difference. The exact aligned integer is
$x=a2^d-b$. Writing $\ell_a=\lfloor\log_2a\rfloor$ and
$\ell_b=\lfloor\log_2b\rfloor$, the compressed branch checks
$d\geq3$ and $\ell_b+2\leq\ell_a+d$, then sets $j=d-3$.
This guard bounds cancellation: $b<2^{\ell_a+d-1}$ while
$a2^d\geq2^{\ell_a+d}$. Consequently $x\geq2^{p+j+2}$.
For the final rounding shift $s=\lfloor\log_2x\rfloor-p$, this supplies
exactly $j+2\leq s$.
The limb calculation must also construct $J_j(x)$ correctly. Decompose
$b=Q2^j+r$, with $0\leq r<2^j$. When $r\neq0$,
\[
 x=(8a-Q-1)2^j+(2^j-r),
 \qquad J_j(x)=(8a-Q-1)\mathbin{\mathrm{OR}}1.
\]
The final one in the subtraction is the borrow shown in
\Cref{fig:sticky}(b). When $r=0$, the result is simply $8a-Q$.
The guard also proves that the limb subtraction cannot underflow.
The representation lemmas identify these limb operations with $J_j(x)$;
\Cref{thm:sticky} then identifies their rounding with rounding the exact $x$.

\begin{figure}[!htb]
\centering
\begin{tikzpicture}[
  font=\small,
  edge/.style={draw=black!65,semithick},
  arr/.style={-{Stealth[length=1.8mm]},semithick},
  title/.style={font=\small\bfseries,anchor=west},
  lab/.style={font=\small,fill=white,inner sep=1.5pt}
]
\node[title] at (0,1.6) {(a) A sticky bit preserves the rounding decision};
\node at (3.0,1.05) {quotient};
\node at (4.65,1.05) {guard};
\node at (8.25,1.05) {lower suffix $r$};
\draw[edge] (5.15,0.85) -- (5.15,0.75) -- (11.35,0.75) -- (11.35,0.85);

\node[anchor=east] at (1.55,0.325) {$x$};
\draw[edge,fill=blue!6] (1.8,0) rectangle (4.2,0.65);
\draw[edge,fill=black!7] (4.2,0) rectangle (5.1,0.65);
\draw[edge] (5.1,0) rectangle (7.85,0.65);
\draw[edge] (7.85,0) rectangle (8.75,0.65);
\draw[edge,fill=orange!12] (8.75,0) rectangle (11.35,0.65);
\node at (3.0,0.325) {$q$};
\node at (4.65,0.325) {$g$};
\node at (6.475,0.325) {$\cdots$};
\node at (8.3,0.325) {$\eta$};
\node at (10.05,0.325) {$v=x\bmod2^j$};

\node[anchor=east] at (1.55,-0.875) {$J_j(x)$};
\draw[edge,fill=blue!6] (1.8,-1.2) rectangle (4.2,-0.55);
\draw[edge,fill=black!7] (4.2,-1.2) rectangle (5.1,-0.55);
\draw[edge] (5.1,-1.2) rectangle (7.85,-0.55);
\draw[edge,fill=orange!20] (7.85,-1.2) rectangle (8.75,-0.55);
\node at (3.0,-0.875) {$q$};
\node at (4.65,-0.875) {$g$};
\node at (6.475,-0.875) {$\cdots$};
\node at (8.3,-0.875) {$\eta'$};
\draw[arr] (3.0,-0.04) -- (3.0,-0.51);
\draw[arr] (4.65,-0.04) -- (4.65,-0.51);
\draw[arr] (6.475,-0.04) -- (6.475,-0.51);
\draw[arr,orange!75!black] (10.05,-0.04) |- (8.8,-0.875);
\node[lab,anchor=west] at (10.18,-0.45) {remove $j$ bits};
\node[anchor=west] at (8.0,-1.57)
  {$\eta'=\eta\mathbin{\mathrm{OR}}\mathbf{1}[v\neq0]$};
\node[anchor=west] at (1.8,-1.57)
  {$j+2\leq s$: the jam stays below $g$.};

\node[title] at (0,-2.35) {(b) Account for the subtraction borrow before jamming};
\node at (3.0,-2.87) {prefix $\times\,2^j$};
\node at (5.8,-2.87) {tail};
\node[anchor=east] at (1.35,-3.425) {$x$};
\draw[edge,fill=blue!6] (1.6,-3.75) rectangle (4.4,-3.1);
\draw[edge,fill=orange!12] (4.4,-3.75) rectangle (7.2,-3.1);
\node at (3.0,-3.425) {$8a-Q$};
\node at (5.8,-3.425) {$-r$};
\draw[arr] (3.0,-3.8) -- node[lab,left,xshift=-3pt] {borrow $1$} (3.0,-4.4);
\draw[arr] (5.8,-3.8) -- node[lab,right,xshift=3pt] {add $2^j$} (5.8,-4.4);
\node[anchor=east] at (1.35,-4.775) {$x$};
\draw[edge,fill=blue!6] (1.6,-5.1) rectangle (4.4,-4.45);
\draw[edge,fill=orange!12] (4.4,-5.1) rectangle (7.2,-4.45);
\node at (3.0,-4.775) {$8a-Q-1$};
\node at (5.8,-4.775) {$2^j-r>0$};
\node[edge,fill=orange!8,minimum height=0.65cm,inner sep=4pt]
  (jammed) at (10.25,-4.775) {$(8a-Q-1)\mathbin{\mathrm{OR}}1$};
\draw[arr] (7.25,-4.775) -- node[lab,above] {jam} (jammed.west);
\node[align=center] at (10.25,-3.55)
  {$b=Q2^j+r$\\$0<r<2^j,\quad j=d-3$};
\end{tikzpicture}
\caption{Information preserved by compressed alignment.
(a) The bit $\eta$ is at position $j$ before the shift and position zero
afterwards. Jamming summarizes the removed tail without changing $q$, $g$,
or whether the lower suffix is zero; groups are schematic and may be empty.
(b) In the guarded subtraction $x=a2^d-b$, a nonzero remainder borrows
from the prefix. Its complemented tail is still nonzero, so the low bit
is then set. If $r=0$, the prefix is simply $8a-Q$.}
\label{fig:sticky}
\end{figure}

The complete subtraction theorem also covers exact alignment and
reference fallbacks outside the compressed or normal-result branches.
Decoding gives model subtraction; the codec inverse supplies the encoded
equality required by \lean{Certified}. Wide-limb FMA first forms an exact
limb product, then applies the same alignment theorem to the product and
addend. Because the theorem needs no upper significand bound, it handles
the longer product without intermediate rounding.
Appendix~\ref{app:execution} gives the complete dispatcher argument.

\subsection{Certified quotient checking}
\label{sec:division}

The divider checks a proposed answer independently of the algorithm
that found it. Given normalized significands $n,d$ with
$p+1$ bits, where $64<p\leq126$, the two-word backend scales the numerator
by $2^s$, choosing $s=p$ if $n\geq d$ and $s=p+1$ otherwise. A radix-$2^{32}$
implementation of Algorithm~D~\cite{knuth1997taocp2} proposes a quotient
$q$ and remainder $r$. The backend accepts them only after independently
checking
\begin{equation}\label{eq:execution-division-certificate}
 qd+r=n2^s,\qquad r<d.
\end{equation}
These relations characterize Euclidean division, so any representable
pair passing the check suffices for the accepting branch.
The check and fallback supply the same relation to the final rounder
(\Cref{fig:division}).

\begin{figure}[!htb]
\centering
\begin{tikzpicture}[
  font=\small,
  box/.style={draw=black!65,rounded corners=1pt,align=center,
    minimum height=0.9cm,inner sep=4pt},
  arr/.style={-{Stealth[length=1.8mm]},semithick},
  lab/.style={font=\small,fill=white,inner sep=2pt}
]
\node[box,text width=1.45cm] (input) at (0.9,0)
  {Inputs\\$n,d,s$};
\node[box,fill=blue!6,text width=2.1cm] (propose) at (3.65,0)
  {Algorithm D\\proposes $(q,r)$};
\node[box,fill=orange!8,text width=2.95cm] (check) at (7.45,0)
  {Check\\$qd+r=n2^s$\\$r<d$};
\node[box,fill=blue!6,text width=2.7cm] (round) at (11.85,0)
  {Round significand\\using $d,q,r$};
\node[box,text width=2.95cm] (restore) at (7.45,-2.0)
  {Proved restoring\\divider};
\draw[arr] (input) -- (propose);
\draw[arr] (propose) -- (check);
\draw[arr] (check) -- node[lab,above] {accept} (round);
\draw[arr] (check) -- node[lab,right] {reject} (restore);
\draw[arr] (input.south) |- node[lab,pos=0.7,above] {original $n,d,s$}
  (restore.west);
\draw[arr] (restore.east) -| node[lab,pos=0.25,above] {exact $(q,r)$}
  (round.south);
\node[font=\small,align=center] at (6.45,-2.92)
  {Both routes use the same rounding rule: compare $2r$ with $d$, then resolve ties by parity.};
\end{tikzpicture}
\caption{Checked division for $2^p\leq n,d<2^{p+1}$,
$64<p\leq126$, and $s=p+\mathbf{1}[n<d]$.
The accepting branch uses the checked equation and remainder bound.
The fallback establishes them by a restoring-division invariant and needs
no second check. The Boolean check runs during execution; its soundness
proof is erased.}
\label{fig:division}
\end{figure}

The runtime check widens a two-word product to four words, adds the
remainder, rejects a final carry, and checks equality with the shifted
numerator and the bound $r<d$. The zero-carry conjunct lets the proof
interpret this as an exact natural-number equation. The 128-bit argument
bounds already give $qd+r<2^{256}$, so the carry test excludes no
representable pair. The precision bound also places $n2^s$ within four
words.

\begin{theorem}[Soundness of the quotient check]\label{thm:quotient}
Let $64<p\leq126$, $n,d,q,r<2^{128}$, and $s\in\{p,p+1\}$.
If the native certificate accepts $(q,r)$ for the shifted numerator
$n2^s$ and divisor $d$, then
\[
 q=\left\lfloor n2^s/d\right\rfloor,
 \qquad r=(n2^s)\bmod d.
\]
\end{theorem}
\begin{proof}
The multiplication and addition refinement theorems recover the exact
natural-number equation from the checked word operations and zero carry.
The comparison theorem gives $r<d$, hence $d>0$. Uniqueness of Euclidean
division then identifies both quotient and remainder.
\end{proof}

A rejected candidate falls back to a proved restoring divider satisfying
the same equation and bound. One correctness theorem therefore covers
both branches. The Boolean check executes at runtime; its soundness
proof is erased.

Once the pair is exact, rounding compares $2r$ with $d$ and increments
$q$ when $2r>d$, or when $2r=d$ and $q$ is odd. Capacity hypotheses prevent
overflow when doubling the remainder or incrementing the quotient.
The native rounding theorem identifies the result with exact
nearest-even quotient rounding, after which exponent handling and
packing complete the backend refinement.

\subsection{A rounding bound for a custom format}
\label{sec:affine-proof}

For the custom 24-bit format of \Cref{sec:models}, we prove both encoded
equality and a real-valued rounding bound for $ax+b$:
\par\noindent\begin{minipage}{\linewidth}
\begin{lstlisting}
import FloatLib
open FloatLib.Floats Formats.BinaryInterchange
open ExecFloat.Binary (toModel isFinite)

abbrev F := ExecFloat.Binary 8 15
noncomputable abbrev val (x : F) : ℝ := Model.toReal (toModel x)
def affine (a x b : F) : F := ExecFloat.fma a x b

theorem affine_encoded (a x b : F) :
    affine a x b = ExecFloat.Spec.fma a x b :=
  ExecFloat.Proof.fma_eq_spec a x b
\end{lstlisting}
\end{minipage}\par
The encoded equation holds for every operand. Decoding it through the
codec gives a one-rounding theorem for finite inputs and output:
\par\noindent\begin{minipage}{\linewidth}
\begin{lstlisting}
theorem affine_rounds (a x b : F)
    (ha : isFinite a = true) (hx : isFinite x = true)
    (hb : isFinite b = true) (hout : isFinite (affine a x b) = true) :
    val (affine a x b) =
      Model.roundAt (ExecFloat.Binary.format 8 15) (val a * val x + val b) := by
  have hmodel : toModel (affine a x b) =
      Model.fma (toModel a) (toModel x) (toModel b) := by
    rw [affine_encoded, Model.Proof.fma_eq_spec]
    exact ExecFloat.ModelCodec.decode_liftTernary _ a x b
  simpa only [val, hmodel] using
    Model.toReal_fma_eq_roundAt _ _ _ (by decide) ha hx hb (hmodel ▸ hout)
\end{lstlisting}
\end{minipage}\par
The proof applies \Cref{thm:arithmetic-bridge}; \lean{decide} discharges
the closed format condition. Its hypotheses allow subnormal inputs and need no
separate finiteness assumption on the product. For the exact real
expression $z=\operatorname{val}_F(a)\operatorname{val}_F(x)
+\operatorname{val}_F(b)$, \Cref{thm:rounding-error} then gives
$|\operatorname{val}_F(\operatorname{affine}(a,x,b))-z|
\leq\frac12\ulp_F(z)$.

\subsection{Parsing and Lean's native types}\label{sec:interoperation}

The configured binary types parse decimal, hexadecimal, dyadic, and
special-value text through the same descriptor model.
\lean{parse} defaults to nearest-even, with options for rounding
direction, exception status, and input limits. We prove that successful
bounded parses agree with the unbounded parser and that changing the
storage carrier preserves the result and status.
Formatting provides exact decimal and hexadecimal output, or fixed and
scientific decimal output at a requested precision. For finite values,
the latter two have proved nearest-even error bounds on the value
denoted by the actual output string.

The released library uses Lean~4.34.0. Binary32/64 interoperation connects
to Lean's logical \lean{Float32} and \lean{Float}
models~\cite{lean-v4340-float-model}.
The bridges cover integer constructors, addition and subtraction on
finite inputs, and square root. Conversion to Lean's packed models
canonicalizes NaN signs and payloads while preserving numeric bits,
including the sign of zero; the square-root theorem accounts for this
canonicalization. For example, the addition bridge transports a
left-associated native sum into certified software arithmetic:
\par\noindent\begin{minipage}{\linewidth}
\begin{lstlisting}
import FloatLib
open FloatLib.Floats.ExecFloat.Binary

theorem import_sum3 (x y z : Float32)
    (hx : x.isFinite = true) (hy : y.isFinite = true)
    (hz : z.isFinite = true) (hxy : (x + y).isFinite = true) :
    ofFloat32 ((x + y) + z) =
      (ofFloat32 x + ofFloat32 y) + ofFloat32 z := by
  rw [ofFloat32_add_of_isFinite (x + y) z hxy hz,
    ofFloat32_add_of_isFinite x y hx hy]
  rfl
\end{lstlisting}
\end{minipage}\par
Only the inputs and first sum must be finite; the final sum may overflow.
The equation uses Lean's logical native model and preserves the two
rounding steps.

\subsection{Compilation and runtime assumptions}

Lean erases proof fields when compiling arithmetic. Proved compiler
simplification equations replace selected reference primitives with
faster equal functions, provided the equation is imported before the
caller is compiled. Closed descriptors and direct entry points expose
representation and selection choices to inlining.

These proofs establish equations between Lean functions, including the
logical models of native operations; they do not verify external machine
instructions. Compiled execution relies on Lean's compiler and runtime.
The benchmarked FloatLib paths use proved software kernels. Separate
\lean{NativeFPU.Unchecked} adapters use host arithmetic without the same
equality certificate. The native-C timings in \Cref{sec:evaluation}
provide a hardware baseline.
\section{Format-specific arithmetic}\label{sec:formats}

Each family connects its representation and reference operations to the
execution certificate, while the numerical guarantees in \Cref{tab:formats}
depend on that family's representation and rounding rules.
The posit proofs justify rounding at encoding boundaries and show how
exact updates compose within a finite quire.

\begin{table}[!htb]
\caption{Numerical families in FloatLib and representative guarantees.
Each family states the domain, range, and encoding conditions appropriate
to its operations.}
\label{tab:formats}
\small
\setlength{\tabcolsep}{3pt}
\begin{tabularx}{\linewidth}{@{}p{1.9cm}XX@{}}
\toprule
Family & Representation & Representative contract \\
\midrule
Binary & Parameterized fields, bias, and exceptional encodings
  & Encoded arithmetic refinement; IEEE-style real rounding and exactness \\
Decimal & Coefficient and quantum; BID and DPD encodings
  & Datum-level arithmetic, rounding error, and status \\
Posit & Regime, two exponent bits; \mbox{$n\geq2$}
  & Standard rounding; exact accumulation in a $16n$-bit quire \\
P3109 & Width, precision, signedness, and domain
  & Destination projection, saturation, and mixed-format operations \\
MX & 32 lanes with a common E8M0 scale
  & Quantization at a selected scale; one-round exact dot product \\
Other codes & Fixed point, logarithmic, codebook, and shared scale
  & Scale-aware arithmetic or finite-set quantization \\
\bottomrule
\end{tabularx}
\end{table}

\subsection{Posit rounding thresholds}
\label{sec:posit-rounding}

An $n$-bit posit uses a variable-length regime, up to two exponent bits,
and the remaining fraction bits, concentrating precision near one.
It has one zero and one exceptional value, NaR (``not a real'').
The 2022 standard fixes the exponent convention and specifies rounding
using the format one bit wider~\cite{posit2022standard}.

Let $D_n(u)$ denote the rational value of a nonnegative $n$-bit code $u$.
For adjacent interior codes $u,u+1$, the rounding threshold is
\begin{equation}\label{eq:posit-threshold}
 T_n(u)=D_{n+1}(2u+1).
\end{equation}
\Needspace{5\baselineskip}
Appending a one to the lower code gives the boundary in the wider
format. At equality, the retained even code wins. Zero is exact;
nonzero values below the least positive value round to that value in
magnitude, and values above the greatest finite value saturate. Negative
results use the corresponding symmetric rule.

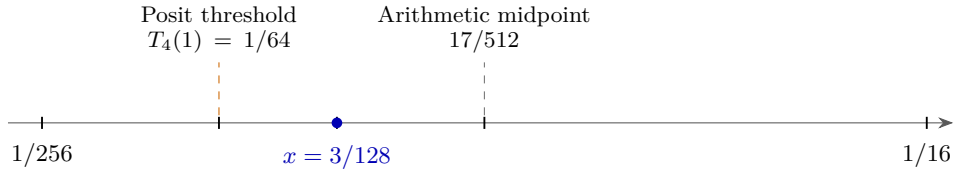
\begin{figure}[t]
\centering
\begin{tikzpicture}[x=1cm,y=1cm,font=\small,
  axis tick/.style={line width=0.7pt},
  note/.style={align=center,font=\footnotesize,text width=2.8cm}]
\draw[-{Stealth[length=2mm]},black!65] (0.15,0) -- (12.65,0);
\foreach \x in {0.6,2.94,4.5,6.45,12.3}
  \draw[axis tick] (\x,-0.08) -- (\x,0.08);
\node[below=4pt] at (0.6,0) {$1/256$};
\node[below=4pt] at (12.3,0) {$1/16$};
\draw[orange!80!black,dashed] (2.94,0.10) -- (2.94,0.80);
\node[note,above] at (2.94,0.80)
  {Posit threshold\\$T_4(1)=1/64$};
\draw[black!55,dashed] (6.45,0.10) -- (6.45,0.80);
\node[note,above] at (6.45,0.80)
  {Arithmetic midpoint\\$17/512$};
\fill[blue!70!black] (4.5,0) circle (2pt);
\node[below=5pt,text=blue!70!black] at (4.5,0) {$x=3/128$};
\end{tikzpicture}
\caption{Two different rounding boundaries for adjacent four-bit posits.
The values are placed on a linear real axis. The standard posit rule
rounds $x$ to $1/16$ because it lies above the appended-code threshold;
absolute nearest-value rounding would choose $1/256$.}
\label{fig:posit-threshold}
\end{figure}

For standard four-bit posits, codes 1 and 2 denote $1/256$ and $1/16$;
five-bit code 3 gives the threshold $1/64$. Thus $x=3/128$ rounds to
code 2 although code 1 has smaller absolute error
(\Cref{fig:posit-threshold}). Writing $W_n$ for the $n$-bit posit words,
we give $\Pi_n:\R\to W_n$ a separate real specification using the standard's
encoding boundaries, and implement it through a search driven by
comparisons with rational thresholds. This separates the finite code
search from the target's comparison method.

\begin{theorem}[Posit rounding from exact comparisons]
\label{thm:posit-comparison}
Let $n\geq2$, $y\in\R$, and let
$C:\Q\to\{<,=,>\}$ satisfy
$C(t)=\operatorname{cmp}(y,t)$ for every rational $t$.
The signed comparator-based search returns $\Pi_n(y)$.
In particular, the rational executable rounder returns $\Pi_n(q)$
on every rational input $q$.
\end{theorem}
\begin{proof}
For an interior positive target, let $c$ be the greatest code whose value
does not exceed it. The bisection invariant is
$\ell\leq c<h$ and $h-\ell\leq2^f$, where $f$ is the number of remaining
steps. Starting from $[0,2^{n-1})$, $n$ steps suffice to isolate $c$.
The comparator equation makes the executable and real searches take
identical branches, finding the same lower code and agreeing at the
threshold \eqref{eq:posit-threshold}, including equality and its parity
test. The zero, minpos, and saturation branches agree by the same
equation. For a negative target, reversing comparisons with negated
thresholds gives the magnitude comparator; whole-word negation restores
the sign. Rational comparison supplies the equation for rational targets.
\end{proof}

\subsection{Exact comparison for elementary functions}
\label{sec:elementary}

For a rational radicand $a\geq0$ and rational threshold $t\geq0$,
$t\leq\sqrt a\Longleftrightarrow t^2\leq a$ supplies an exact comparison.
Negative thresholds lie below the root and are decided directly.
The equivalence and its strict counterpart determine the code and threshold
tests in \Cref{thm:posit-comparison}, proving standard posit rounding of
the exact root, including ties and extreme intervals.

Exponential and logarithm use rational enclosures $I_k=[L_k,U_k]$ of the
real target $y$, at approximation degrees $2^k$. Refinement continues until
$U_k<t$ or $t<L_k$ decides the comparison with a rational threshold $t$
(\Cref{fig:posit-enclosure}).
Containment makes each reported ordering sound; convergence of both
endpoints to $y$ ensures eventual separation when $y\neq t$.
The equalities $\exp(0)=1$ and $\log(1)=0$ are handled directly.
For rational $q\neq0$, $\exp q$ is
irrational, as is $\log q$ for $q>0$, $q\neq1$, excluding equality with
rational thresholds in the remaining cases. The comparator supplies
both code-value queries and the final threshold test in
\Cref{thm:posit-comparison}: for an ordinary word $v$ denoting $q$,
\[
 \operatorname{exp}_n(v)=\Pi_n(\exp q),
 \qquad
 \operatorname{log}_n(v)=\Pi_n(\log q)\quad(q>0).
\]
Rational powers combine logarithm enclosures with exact integer-power
comparisons; the latter resolve inconclusive enclosures and decide
equality. Roots, trigonometric functions, and hyperbolic functions have
corresponding domain-specific comparison proofs.

\begin{figure}[!htb]
\centering
\begin{tikzpicture}[x=1cm,y=1cm,font=\small,
  enclosure/.style={draw=blue!70!black,line width=1.2pt},
  bound/.style={text=blue!70!black,below=4pt},
  stage/.style={anchor=west,align=left,inner sep=0pt},
  query/.style={draw=black!55,rounded corners=1.5pt,
    align=center,inner sep=5pt,minimum height=0.75cm}]
\node[stage] at (0,0) {Unresolved\\$t\in[L_k,U_k]$};
\node[stage] at (0,-1.65) {Separated ($j>k$)\\$U_j<t$};
\draw[-{Stealth[length=1.8mm]},black!65]
  (0.6,-0.50) -- node[right=3pt] {refine} (0.6,-1.10);

\foreach \h in {0,-1.65}
  \draw[-{Stealth[length=1.8mm]},black!45]
    (2.9,\h) -- (9.25,\h);

\draw[enclosure] (3.9,0) -- (8.5,0);
\foreach \x in {3.9,8.5}
  \draw[enclosure] (\x,-0.11) -- (\x,0.11);
\node[bound] at (3.9,0) {$L_k$};
\node[bound] at (8.5,0) {$U_k$};

\draw[enclosure] (3.35,-1.65) -- (5.75,-1.65);
\foreach \x in {3.35,5.75}
  \draw[enclosure] (\x,-1.76) -- (\x,-1.54);
\node[bound] at (3.35,-1.65) {$L_j$};
\node[bound] at (5.75,-1.65) {$U_j$};

\foreach \h in {0,-1.65}{
  \fill[blue!70!black] (5,\h) circle (1.9pt);
  \node[above=4pt,text=blue!70!black] at (5,\h) {$y$};
}
\draw[orange!80!black,dashed,line width=0.8pt]
  (6.7,0.45) -- (6.7,-2.02);
\node[above=3pt,text=orange!80!black] at (6.7,0.45)
  {$t=T_n(u)$};

\node[query] (search) at (12.6,-1.65) {Posit\\code search};
\draw[-{Stealth[length=1.8mm]},black!65]
  (9.45,-1.65) -- node[above=3pt,text=black] {$C(t)={<}$} (search.west);
\end{tikzpicture}\hspace*{\fontdimen2\font}
\caption{Enclosure comparison at an interior posit threshold
$t=T_n(u)$ from \eqref{eq:posit-threshold}. Both rational intervals
contain $y$. The first test is unresolved; the later bound $U_j<t$
certifies $y<t$ for the code search. The stages are schematic; neither
nesting nor a convergence rate is assumed.}
\label{fig:posit-enclosure}
\end{figure}
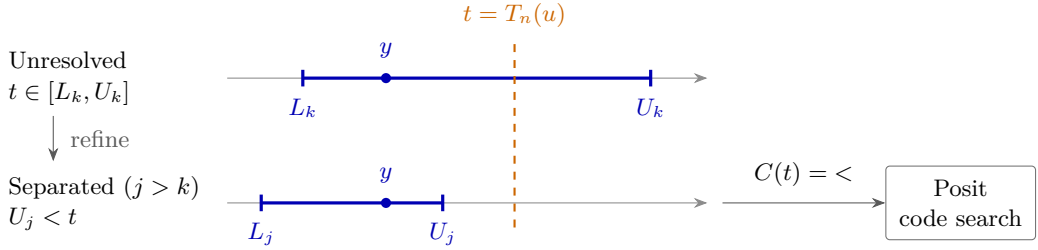

The optional binary elementary-function kernels compute deterministic
approximations; their equality certificates do not establish correct
rounding of the real function. Separate rational exponential
and logarithm enclosures provide proved containment bounds.
Appendix~\ref{app:rounding} proves separation and describes enclosure reuse
and logarithmic comparison for large exponential arguments.

\subsection{Exact quire accumulation}
\label{sec:quire}

The standard quire for an $n$-bit posit is a signed
$16n$-bit accumulator with unit $\Delta_n=2^{16-8n}$.
Every product of two ordinary $n$-bit posits is an integer multiple of
this unit, so products and partial sums remain exact while the accumulator
stays in range. Conversion to a posit rounds the final sum.
Write $\operatorname{val}_n(v)$ and $\operatorname{Q}(A)$ for the rational
values of an ordinary posit and quire, and let
$\operatorname{foldMulAdd}(0,\mathcal P)$ fold the model's exact product
update from zero over operand pairs $\mathcal P$.

\Needspace{8\baselineskip}
\begin{theorem}[Exact quire accumulation]\label{thm:quire}
For a standard posit width $n\geq2$, let
$\mathcal P=[(a_0,b_0),\ldots,(a_{k-1},b_{k-1})]$ contain no NaR operands.
If $k<2^{31}$, then the fold returns an ordinary quire and
\[
 \operatorname{Q}(\operatorname{foldMulAdd}(0,\mathcal P))
     =\sum_{i<k}\operatorname{val}_n(a_i)\operatorname{val}_n(b_i).
\]
\end{theorem}
\begin{proof}
The largest posit magnitude is $2^{4(n-2)}$, so a product has magnitude
at most $2^{8n-16}$, or $B_n=2^{16n-32}$ in quire units.
We induct on $j$, maintaining exactness and $|A_j|\leq jB_n$ for the
quire coefficient. The base case is $A_0=0$. If the next product has
coefficient $C_j$, then
$|A_j+C_j|\leq |A_j|+|C_j|\leq(j+1)B_n$.
For every prefix $j+1\leq k<2^{31}$, this bound is strictly below
$2^{16n-1}$. The next coefficient is therefore representable and distinct
from the reserved NaR word, so the one-step arithmetic theorem preserves
exactness and completes the induction.
\end{proof}

The bound is uniform in $n$ because quire capacity and product bounds
grow together; it holds in every ordering without relying on cancellation.
Longer sums can succeed when their actual partial sums remain in range.

The configured update decodes operands through their codec and calls the
model update. The codec laws transfer the fold result to
\lean{toRat?\_foldl\_qMulAdd\_zero}, whose successful rational decoding
also excludes NaR. For the resulting quire $A$, the model conversion
theorem and codec round trip give
\[
 \operatorname{toModel}(\operatorname{qToP}(A))
   =\operatorname{roundRat}_n
       \left(\sum_{i<k}\operatorname{val}_n(a_i)
                         \operatorname{val}_n(b_i)\right).
\]
By \Cref{thm:posit-comparison}, this is one standard posit rounding of the
exact dot product.

\subsection{Binary, decimal, and scaled formats}

\para{Low-precision binary.}
Small binary formats reuse field descriptors with different exceptional
encodings~\cite{onnx-float8,ocp-mx-2023}. E5M2 and bfloat16 use IEEE
exceptional conventions; E4M3FN has NaN patterns but no infinity.
FNUZ reserves the negative-zero word for its unique NaN and changes the
bias. MX's FP4 and FP6 element formats have neither NaN nor infinity.
The reference arithmetic applies the policy when packing results.
For encodings in which every word is finite, such as FP4 and FP6,
directed add, subtract, and multiply bounds require the exact result to
lie within the finite range; these theorems do not cover the NaN-bearing
FN and FNUZ policies.

A binary operation on byte-sized values admits a complete table under
the same execution certificate. An eight-bit FMA table would require
$2^{24}$ entries; the released FMA implementation uses the exact baseline.

\para{Decimal arithmetic.}
Decimal datums retain the quantum exponent $q$ as well as the value
$(-1)^s c\,10^q$: $1.50$ and $1.5$ are distinct, so BID and DPD encode
both coefficient and quantum exponent~\cite{ieee754-2019}. FloatLib supports
decimal32, decimal64, decimal128, and parameterized layouts with
$3d+1$ digits, positive exponent-continuation width, and arbitrary integer
bias. Arithmetic contracts preserve datum validity and status;
nearest-even and nearest-away error bounds are half the selected decimal
grid unit for nonoverflowing results. Division requires a nonzero divisor;
square root requires a nonnegative finite input and retains the
no-overflow premise for general descriptors.

\para{P3109.}
P3109 operations follow working-group report 4.0.3~\cite{p3109}, with
independent source and destination descriptors. They evaluate finite inputs
exactly, then project into the destination under explicit rounding and
saturation choices. Deterministic modes have direction and error results
before saturation; the full projection theorem includes saturation.
Stochastic modes take random bits explicitly and establish grid membership,
which alone does not imply unbiasedness.

\para{Microscaling.}
MX couples 32 lanes through a shared E8M0 scale
\cite{ocp-mx-2023}. At a selected finite scale, saturating conversion
minimizes each lane's error among finite encodings; the contract does not
assert a global optimum over scales. For finite elements and scales,
we choose exact product accumulation followed by one binary32 rounding
for the MX dot product.

\para{Fixed point and other encodings.}
Fixed-point, logarithmic, codebook, and generic shared-scale values use
the representation interface of \Cref{sec:representations}.
Multiplying fixed-point codes at scales $d_1,d_2$ produces scale
$d_1+d_2$; codebook quantization minimizes distance over the table's
finite entries when at least one exists. Each family states its own
scale and range conditions.
\section{Evaluation}\label{sec:evaluation}

On the fixed scalar workload, FloatLib has lower medians than Universal
in 34 of 65 matched posit cases, including all five operations at 64 bits
and all twenty cases at widths 512--4096. The 31 slower cases include
all twenty at widths 5--8. All 84 binary medians exceed MPFR's.
Separate Lean comparisons measure P3109 arithmetic and scalar conversions
(\Cref{sec:lean-comparisons}); independent numerical checks follow.

\para{Formats and execution.}
We time addition, subtraction, multiplication, division, square root,
and FMA at binary widths
$4,5,6,7,8,16,32,64,128,256,512,1024,2048,4096$.
Widths 4--8 use E2M1, E2M2, E3M2, E3M3, and E4M3: E counts exponent
bits and M stored fraction bits. Widths 16--128 use IEEE interchange
layouts; larger widths use 19 exponent bits and $w-20$ fraction bits,
hence significand precision $P=w-19$ including the leading bit.
Posits add widths 2 and 3, with $es=2$ throughout. Their precision varies
with magnitude: a 32-bit posit can have 28 significand bits, against
binary32's 24. Equal encoded width does not imply equal accuracy.

Binary32/64 call proved word kernels directly; other binary widths and
all posits use public throughput-policy dispatch. For binary addition,
\path{backend-regimes.csv} records the paths described in
\Cref{sec:execution}: exhaustive tables at 4--8 bits, word kernels at
16--64, an internal fixed-limb kernel at 128, and the exact baseline at
256--4096. Wide rows use the default carrier; the optional limb-array
carrier was not measured.

\para{Workload and measurement.}
Each operation uses sixteen tuples rounded before timing
(\Cref{app:evaluation-details}). A result fingerprint selects the next
tuple, creating a dependency without accumulating arithmetic results.
Timing includes fixture selection, arithmetic, representation handling,
and result observation; construction, compilation, and proof checking
are excluded. Sources, trials, and environment details accompany the
\href{https://github.com/lean-dojo/FloatLib/tree/main/benchmarks/results}{evaluation data}.

Measurements use an Intel Xeon Platinum 8488C with affinity to one
logical CPU; affinity does not reserve the physical
core. C adapters use \texttt{-O3 -march=native} without fast-math;
SoftFloat uses \texttt{-O2}. Each cell has nine fresh-process trials,
with 256 warmup iterations for plotted implementations. A pilot targets
200\,ms; accepted measurements exceed 50\,ms. The record contains
448 cells and 4,032 observations. We report medians of elapsed time per
iteration and linearly interpolated empirical 5th--95th percentiles
across trials, not confidence intervals.

\para{Reference configurations.}
MPFR~4.2.0 rounds exact rational fixtures nearest-even
\cite{fousse2007mpfr}, with $\mathit{emin}=2-B-(P-1)$ and
$\mathit{emax}=B+1$ for significand precision $P$ and exponent bias $B$.
Every conversion and arithmetic operation passes its ternary rounding
result to \texttt{mpfr\_subnormalize} to match gradual
underflow~\cite{mpfr-manual}.
SoftFloat revision \texttt{a0c6494}, with the \texttt{8086-SSE} specialization,
and native C cover binary32/64 \cite{hauser-softfloat}.
Native C checks representations, rounding mode, and fixture outputs
against MPFR. Binary fingerprints and fixture traces agree on an untimed
256-step prefix; independently calibrated timed loops can differ in length.
Native C supplies a loop baseline, not instruction latency.

Universal revision \texttt{26e69f5} receives the same encoded posit inputs
and checks all sixteen fixture results \cite{omtzigt2018universal}.
Of 78 operation/width checks, 71 pass; the seven failures are square roots
at widths 64--4096. The six passing square-root configurations use host
binary64 and are also excluded, leaving 65 software comparisons.
We exclude the Flocq~4.2.2 wrapper's 84 timing cells from both plots:
unmatched exponent bounds and separately rounded fixture numerators and
denominators can change the dependent sequence
(\Cref{app:evaluation-details}). This concerns the wrapper, not Flocq's
arithmetic theorems \cite{boldo2011flocq,boldo2017flocq}.

\para{Arithmetic performance.}
\Cref{fig:performance} shows absolute costs across widths.
Binary32 addition takes 138.6\,ns in FloatLib, against 34.1\,ns in MPFR
and 22.9\,ns in SoftFloat; all six binary32 medians exceed both references.
At binary64, FloatLib addition rises to 312.8\,ns and division from
123.7 to 1,447.3\,ns. The overhead varies by operation and width.

\begin{figure}[!htb]
\centering
\includegraphics[width=\linewidth]{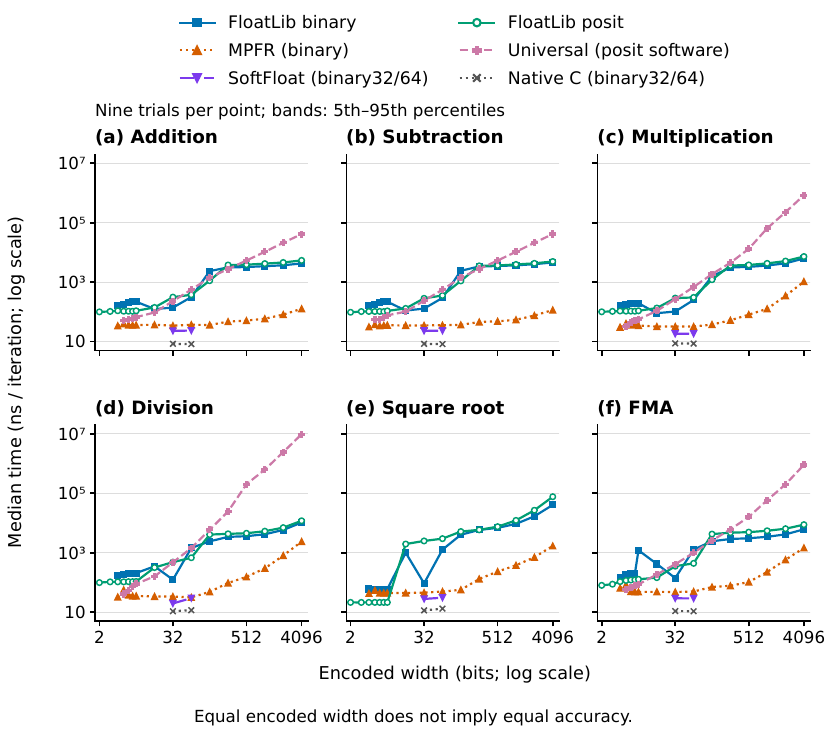}
\caption{Absolute scalar costs across encoded widths. Markers are medians
of nine trials; bands show empirical 5th--95th percentiles. Both axes are
logarithmic. Timing includes input selection and result observation.
MPFR matches binary precision and range; Universal uses the same posit
words and $es=2$. Universal host-assisted square roots and the unmatched
Flocq wrapper are excluded. Equal width does not imply equal accuracy.}
\label{fig:performance}
\end{figure}

\Cref{fig:performance-ratios} compares costs within each family.
Posit multiplication takes 288.9 versus 266.5\,ns at 32 bits,
305.1 versus 684.9\,ns at 64 bits (Universal/FloatLib $2.25$), and
7.25 versus 839.63\,$\mu$s at 4096 bits.
The accompanying data give all 149 matched ratios and trials.
Fixed fixtures and selected paths do not characterize arbitrary
cancellation, exponent gaps, exceptional inputs, or application performance,
nor isolate allocation, representation, and arithmetic costs.

\begin{figure}[!htb]
\centering
\includegraphics[width=\linewidth]{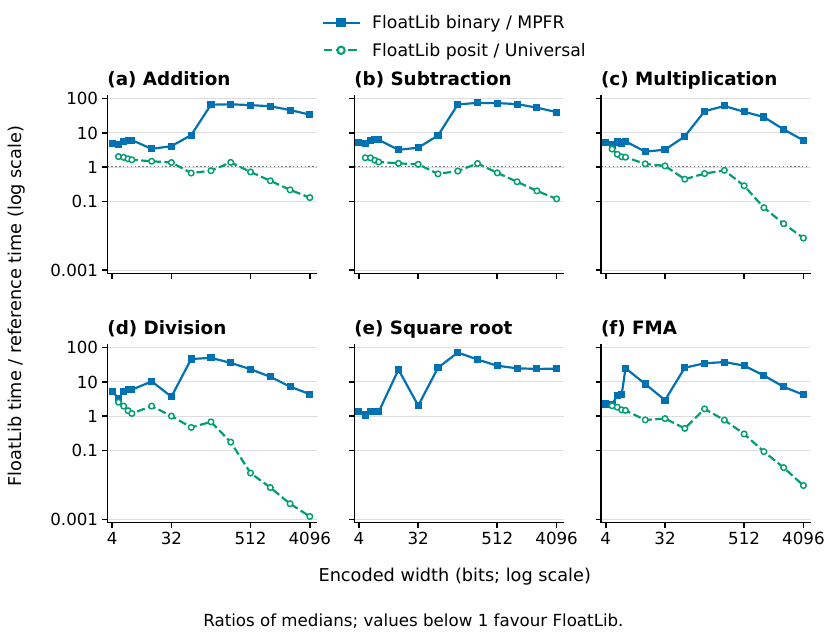}
\caption{Relative costs within each family: FloatLib binary/MPFR and
FloatLib posit/Universal software. Each point divides two nine-trial
medians; values below one favour FloatLib. Both axes are logarithmic.
The 149 ratios comprise 84 binary and 65 posit comparisons; Universal
square roots are excluded as explained above. Lines connect measured
widths; ratios have no uncertainty bands.}
\label{fig:performance-ratios}
\end{figure}

\subsection{Comparisons with Lean libraries}\label{sec:lean-comparisons}

Separate experiments compare P3109 arithmetic with FLoPS and scalar
low-precision conversions with
TensorLib~\cite{flops2026artifact,tensorlib2026}.
On a Xeon Platinum 8275CL with affinity to one logical CPU, each case has
nine paired fresh-process trials in alternating order, 256 warmup
iterations, and sixteen shared fixtures. Results select subsequent inputs;
timing includes lookup, the operation, encoded output, and checksums,
excluding input construction, startup, and file output.
The libraries use their respective supported toolchains; the ratios
include compiler differences.

\para{P3109 arithmetic.}
FloatLib has lower medians in nine of twelve cases
(\Cref{tab:p3109-flops}); FLoPS is faster for four-bit addition,
multiplication, and FMA. FLoPS division costs $1.32$--$1.46$ times
FloatLib's median. Eight-bit FMA differences are small: the empirical
5th--95th percentile intervals overlap for \lean{Binary8p3se}.
Both execute software arithmetic: FloatLib uses exact rationals, FLoPS
dyadics and quotient-based rounding. These timings do not isolate
representation costs.

\begin{table}[!htb]
\centering
\small\setlength{\tabcolsep}{5pt}
\caption{P3109 arithmetic: median $\mu$s per dependent-loop iteration across nine trials on a Xeon 8275CL. Inputs and nearest-even, nonsaturating policies match; the libraries use their respective Lean toolchains.}
\label{tab:p3109-flops}
\begin{tabular}{@{}llrrrr@{}}
\toprule
Format & Library & Add & Multiply & Divide & FMA \\
\midrule
Binary4p2sf & FloatLib & 4.37 & 4.47 & 4.32 & 5.86 \\
Binary4p2sf & FLoPS & 4.29 & 4.13 & 5.71 & 5.16 \\
\addlinespace[2pt]
Binary8p4se & FloatLib & 4.89 & 4.96 & 4.94 & 6.40 \\
Binary8p4se & FLoPS & 5.57 & 5.42 & 7.24 & 6.66 \\
\addlinespace[2pt]
Binary8p3se & FloatLib & 4.79 & 4.99 & 4.91 & 6.31 \\
Binary8p3se & FLoPS & 5.43 & 5.29 & 7.00 & 6.40 \\
\bottomrule
\end{tabular}
\end{table}

The formats are \lean{Binary4p2sf}, \lean{Binary8p4se}, and
\lean{Binary8p3se}, with nearest-even rounding and no saturation.
Addition, multiplication, and division agree on every ordered encoded
pair, including exceptional values. FMA agrees on every four-bit triple
and, for each eight-bit format, every encoded pair $(x,y)$ with third
operand code $(17x+29y+43)\bmod256$. Together these checks compare
529,152 results with no bit differences; all 192 timing fixtures also agree.

\para{Low-precision conversion.}
FloatLib/TensorLib median ratios range from $5.9$ to $13.0$
(\Cref{tab:tensorlib}). All sixteen timed inputs are normal and exactly
representable; subnormal and exceptional cases are unmeasured.
TensorLib specializes normal conversions using bit-field operations;
FloatLib uses shared format-descriptor casts. Tensor arithmetic uses
host floats (\Cref{sec:related}).

\begin{table}[!htb]
\centering
\small\setlength{\tabcolsep}{5pt}
\caption{Low-precision conversion: median ns per dependent-loop iteration across nine trials on a Xeon 8275CL. Encoding converts binary32 to the listed format; decoding converts back. The sixteen timed inputs are normal finite values.}
\label{tab:tensorlib}
\begin{tabular}{@{}lrrrr@{}}
\toprule
 & \multicolumn{2}{c}{Encode} & \multicolumn{2}{c}{Decode} \\
\cmidrule(lr){2-3}\cmidrule(l){4-5}
Format & FloatLib & TensorLib & FloatLib & TensorLib \\
\midrule
FP16 & 2201.0 & 265.3 & 2098.2 & 268.9 \\
bfloat16 & 2219.9 & 252.6 & 1590.0 & 268.3 \\
E4M3FN & 3642.9 & 279.7 & 2187.5 & 272.9 \\
E5M2 & 2248.1 & 267.9 & 2083.1 & 270.5 \\
\bottomrule
\end{tabular}
\end{table}

\Needspace{6\baselineskip}
Conversions connect binary32 with FP16, bfloat16, E4M3FN, and E5M2.
Decoding checks all destination encodings; encoding checks finite
destination values, adjacent midpoints and binary32 neighbours, both signs,
exponent and overflow boundaries, NaNs, and 8,192 seeded random words per
format. Both libraries match an independent nearest-even reference in all
952,612 cases, comparing non-NaN bits and NaN class.
The 2,964 E4M3FN NaN-sign differences occur on overflow and infinity:
FloatLib returns \texttt{7f}, TensorLib \texttt{ff}.
Status flags and NaN payloads are not compared.

\subsection{Independent numerical checks}

\Cref{tab:evaluation-comparisons} distinguishes counts and comparison
relations. The TestFloat adapter calls logical \texttt{Model.*WithStatus}
operations, not every timed backend. It checks non-NaN words and five
flags; NaNs agree by quiet/signalling class, ignoring payload and sign
\cite{hauser-testfloat}. Arithmetic and conversion cases use four rounding
directions and tininess after rounding, excluding binary80, nearest-away,
round-to-odd, and tininess-before-rounding.
FMA accounts for 98,131,968 cases, or 95.78\% of the TestFloat total.

\begin{table}[!htb]
\centering
\small\setlength{\tabcolsep}{3pt}
\caption{Independent numerical comparisons. Counts include repeated
evaluations; zero denotes agreement under the stated relation.
Combined SoftPosit counts include the initial comparison.
Decimal evidence consists only of summaries.}
\label{tab:evaluation-comparisons}
\begin{tabular}{@{}>{\raggedright\arraybackslash}p{.17\linewidth}
  >{\raggedright\arraybackslash}p{.39\linewidth}
  r>{\raggedright\arraybackslash}p{.19\linewidth}@{}}
\toprule
Reference & Tested property & Evaluations & Differences \\
\midrule
TestFloat & Binary model words and five flags & 102,454,320 & 0 \\
MPFR & Primitive rounding; exact reductions & $400+672$ & 0 \\
MPFR, binary16 & Two fixed-first-operand rows & 131,072 & 0 \\
IBM FPgen & Selected binary32 results and flags & 81,513 & 0 \\
Z3 QF\_FP & Result bits; NaN class & 27,492 & 0 \\
Decimal references & Complete datum and five flags & 13,980 & 0 reported \\
SoftPosit, initial & Public posit results/comparisons & 23,718,760 & 7,562; 5 inputs \\
SoftPosit, combined & Public posit results/comparisons & 32,893,800 & 7,564; 7 inputs \\
\bottomrule
\end{tabular}
\end{table}

MPFR correctness checks use version~4.1.0. The binary16 rows enumerate
second operands for negative-zero addition and signalling-NaN
multiplication, not all operand pairs. IBM admits 81,513 of 130,471
records after excluding decimal, missing-result, trap, and incompatible
NaN/underflow cases. Z3 compares 25,168 results by bits and 2,324 NaNs by
class. Separate ONNX and P3109 table checks concern decoding only.
Counts and whole-suite runtimes supply no arithmetic speed ratio.

Decimal summaries report 10,425 arithmetic comparisons, 2,775 square
roots bounded by outward MPFR intervals and rounded independently with
libmpdec, and 780 exact or special-value cases: 13,980 total.
They observe sign, coefficient, quantum exponent, special values, and
five flags, distinguishing cohorts such as $1.50$ and $1.5$.
Multiple NaNs follow the documented first-NaN choice; unspecified
oracle-generated NaN signs are normalized.
The released artifact includes summary reports for these comparisons,
but not the raw vectors, runners, or tested source revision needed to
replay them.

\para{SoftPosit discrepancies.}
SoftPosit revision \texttt{17d5628185b3} supplies ten arithmetic and
comparison operations. Initial checks exhaust word tuples at widths
2--8 and sample widths 16 and 32. Sampling widths 9--15 adds 9,175,040
evaluations and two differences. Excluding repeated 16/32-bit cells gives
32,893,800 evaluations and 7,564 differences across seven distinct
input tuples; 7,560 occurrences repeat three 32-bit inputs.

Four distinct inputs involve FMA. At width 14,
$a=\mathtt{09f4}$, $b=\mathtt{1f60}$, and $c=\mathtt{0001}$ decode to
$253/2048$, $2^{19}$, and $2^{-48}$.
Thus $ab+c=64768+2^{-48}$. Adjacent words \texttt{1efe} and
\texttt{1eff} represent 64512 and 65024, with appended-bit rounding
threshold 64768, also their midpoint here \cite{posit2022standard}.
Their errors are $256+2^{-48}$ and $256-2^{-48}$: FloatLib correctly
chooses the odd upper word; SoftPosit chooses the even lower word.
Its generic FMA records discarded bits, then overwrites that sticky
information at final rounding, turning a non-tie into a tie.
This illustrates the obligation in \Cref{sec:math-sticky}: the
discarded suffix's nonzeroness must survive until rounding.

The other three inputs concern generic 32-bit addition, multiplication,
and division. Packing has one exponent bit left, so exponent value two
must contribute its upper bit, $2\mathbin{\texttt{>>}}1=1$.
SoftPosit instead shifts by $28-29=-1$, undefined in C/C++.
Exact decoding and the standard's thresholds select FloatLib's results;
SoftPosit's named \texttt{p32} operations agree.
On all seven inputs, public FloatLib operations agree with its exact
specification. Recorded source probes confirm the two C diagnoses on
these inputs, without proving a repaired implementation correct.
\section{Related work}\label{sec:related}

\para{Shared numerical foundations.}
Flocq is our principal mathematical precedent. Boldo and Melquiond's
radix-and-exponent abstraction develops representability, rounding, and
error results once for several numerical
formats~\cite{boldo2011flocq,boldo2017flocq}. We follow it and its
separation of generic rounding from concrete arithmetic.
Flocq~4.2.2 also provides verified executable binary operators, including
exceptional values, bit encodings, the four arithmetic operations, square
root, and FMA, parameterized by precision and exponent
range~\cite{flocq422}. FloatLib adds representation and quantization
interfaces without a radix assumption (\Cref{sec:representations}).
Its alternative software kernels share an encoded-result specification,
so numerical proofs remain valid when storage or algorithms change.

FloatSpec brings Flocq-style theory, IEEE models, and Hoare-style
specifications to Lean~\cite{floatspec2026}. Its
\lean{Calc.Operations.Fplus} aligns integer significands and adds them;
\lean{F2R\_plus} proves that the real value is the exact sum.
The rounded IEEE interfaces, including \lean{Bplus}, \lean{Bmult},
\lean{Bdiv}, and \lean{Bsqrt}, are noncomputable definitions in the cited
version. FloatSpec also proves bridges between its binary64 operations
and Lean's logical \lean{Float.Model}; these connect specifications
without verifying native machine instructions.
Timing the executable \lean{Fplus} in place of a rounded operation would
omit the destination rounding and packing measured for FloatLib.
FloatLib carries the rounded encoded specification through executable
backend selection.

\para{P3109 semantics, algorithms, and execution.}
FLoPS, by Chang, Park, Lim, and Nagarakatte, is the closest related Lean
library for P3109 arithmetic~\cite{chang2026flops,flops2026artifact}.
It connects bits to algebraic values and extended-real semantics with
NaNs, with explicit rounding, saturation, and finite-domain policies.
Its executable artifact proves refinement for addition, multiplication,
division, and FMA: \lean{add\_refines}, for example, equates the decoded
result with projection of the exact semantic sum~\cite{flops2026artifact}.
FloatLib likewise connects executable arithmetic to numerical semantics,
and provides same-format and mixed-format P3109 operations through its
common execution interfaces.

FLoPS also studies what familiar numerical algorithms retain when the
format changes. Its FastTwoSum analysis states the conditions under which
error-free decomposition survives saturation. Its ExtractScalar analysis
shows that a magnitude bound on the extracted leading part can fail at
precision one, where tie behavior differs from the usual even-significand
rule~\cite{chang2026flops}.
Its stochastic-rounding theorems give finite-randomness bias formulas and
root-mean-square bounds for sums of rounding errors under its sampling
assumptions~\cite{flops2026artifact}. FloatLib's explicit-random-bit
refinement and grid-membership results establish the implemented choice
and its admissible outputs; they do not supply that probabilistic analysis.

The specifications follow different working-group reports: FLoPS uses
3.2.1 and FloatLib uses 4.0.3~\cite{chang2026flops,p3109}.
We compare their common nearest-even, nonsaturating operations on
matching formats.
\Cref{sec:lean-comparisons} reports the numerical checks and matched
timings.

\para{Native arithmetic and tensor storage.}
Primitive Floats exposes hardware binary64 arithmetic in Coq through
axioms connecting primitive operations to their logical
specification~\cite{bertholon2019primfloat}.
Martin-Dorel, Melquiond, and Roux connect these primitives to Flocq and
use them in CoqInterval and ValidSDP~\cite{martindorel2023enabling}.
FloatLib's opt-in \lean{NativeFPU.Unchecked} path similarly uses host
binary32/64 operations, without the software-refinement certificate
carried by our certified kernels.

TensorLib addresses the needs of tensor programs: dtypes, byte storage,
shapes, strides, indexing, and broadcasting~\cite{tensorlib2026}.
Its low-precision scalar arithmetic commonly decodes an FP8, FP16, or
bfloat16 value into \lean{Float32}, performs the native operation, and
encodes the result; binary32/64 use Lean's native floating-point types
directly. FloatLib's certified execution instead relates software kernels
to their packed specifications.
The scalar-conversion comparison in \Cref{sec:lean-comparisons} times
normal-value integer bit operations, so it does not measure an FPU
advantage.

\para{From arithmetic to expression and program proofs.}
Gappa combines interval reasoning, rewriting, and forward error analysis
to bound expressions containing exact and rounded operators, producing
proof certificates checkable in Coq~\cite{daumas2010gappa}.
VCFloat2 automates roundoff analysis through reification and proof by
reflection; its specifications accommodate user-defined operations and
nonstandard formats, and its bounds can be combined with VST program
proofs~\cite{appel2024vcfloat2}.
Fluctuat uses interval and affine-set abstractions to track real values,
finite-precision values, and error sources through programs with loops.
It detects unstable tests, but the error analysis in the cited work
assumes matching real and finite-precision control
flow~\cite{goubault2011fluctuat}.
These systems propagate local arithmetic facts through larger
computations. FloatLib supplies such facts for its numerical families;
it does not include an automatic whole-program roundoff analysis.

SymFPU implements SMT floating-point operations using bit-vector operations,
with interfaces that admit concrete execution or symbolic encodings for
bit-blasting~\cite{brain2019symfpu}.
SymFPU connects floating-point semantics to a solver's bit vectors;
FloatLib connects packed operations to Lean implementations and numerical
theorems. Our Z3 comparisons test selected SMT floating-point cases.

For verified interval computation, the desired conclusion is containment.
Melquiond develops verified software arithmetic for numerical proofs in
Coq~\cite{melquiond2012interval}; later work shows how sufficient enclosure
properties can replace correctly rounded directed
operations~\cite{martindorel2023enabling}.
Faissole, Geneau de Lamarlière, and Melquiond verify an optimized
exponential enclosure, including its tables and bit manipulations, and
integrate it into CoqInterval~\cite{faissole2024approximation}.
Containment is a different contract from correct scalar rounding.
FloatLib similarly distinguishes its proved rational exp/log enclosures
from executable binary transcendental approximations. The interval
studies' whole-proof timings and our scalar-operation timings measure
different uses of verified arithmetic.

\para{Connections within Lean.}
FloatLib connects to Lean's logical \lean{Float} and \lean{Float32}
models~\cite{lean-v4340-float-model}. The bridges cover integer
constructors, finite-input addition and subtraction, and square root
with NaN canonicalization (\Cref{sec:interoperation}). They transfer
numerical facts through the logical models while retaining the
assumptions of native execution.
DSLean translates external domain-specific representations into typed Lean terms.
Its Gappa case study reconstructs proofs for exact real interval
arithmetic~\cite{rowney2026dslean}.
Supporting rounded operations in that translation also requires the
corresponding arithmetic theorems.
\section{Conclusion}\label{sec:conclusion}

FloatLib provides arbitrary-precision floating-point arithmetic in Lean
with user-defined formats and rounding rules. Its certified execution
interface connects optimized software kernels to complete encoded
specifications, with numerical theorems establishing rounding, error
bounds, and exactness under their stated assumptions. Shared rounding
invariants and quotient checks allow implementations to avoid unnecessary
computation without changing their results. Proofs about a numerical
program can therefore use the same arithmetic specification as its
implementation is optimized.

The library, proofs, benchmarks, evaluation data, and guide are available
as open source at
\url{https://github.com/lean-dojo/FloatLib}.

\bibliography{refs}
\label{paper:main-end}

\appendix
\Needspace{8\baselineskip}
\section{Execution refinement invariants}
\label{app:execution}

The execution proof in \Cref{sec:execution} proceeds from integer and limb
invariants through packing to complete encoded operations.
\Cref{tab:execution-declarations} groups the Lean declarations by these
proof stages. We retain the notation $J_j$ for jamming and $E_s$
for nearest-even integer rounding. The parameter $p$ denotes stored
fraction width, so normal significands have $p+1$ bits.

\begin{table}[!tb]
\centering
\small
\setlength{\tabcolsep}{4pt}
\renewcommand{\arraystretch}{1.08}
\caption{Execution proof stages and their Lean declarations.
Names are relative to the namespace in each shaded heading.}
\label{tab:execution-declarations}
\begin{tabularx}{\linewidth}{@{}
  >{\raggedright\arraybackslash}p{.29\linewidth}
  >{\raggedright\arraybackslash}X@{}}
\toprule
\textbf{Proof step} & \textbf{Declaration} \\
\midrule
\rowcolor{black!6}
\multicolumn{2}{@{}l@{}}{\strut\textit{Shared invariants}\quad
  \path{FloatLib.Numerics}} \\
Preserve rounding after jamming &
  \path{roundShiftRightEven_shiftRightJam} \\
Account for the subtraction borrow &
  \path{shiftRightJam_mul_two_pow_sub}\newline
  \path{LimbArray.subLoop_spec} \\
\midrule
\rowcolor{black!6}
\multicolumn{2}{@{}l@{}}{\strut\textit{Wide-limb arithmetic}\quad
  \path{FloatLib.Floats.Formats.BinaryInterchange.Model.WideLimb}} \\
Normalize and pack the result &
  \path{round_eq_of_roundJammed}\newline
  \path{roundNormal?_map_toModel} \\
Combine alignment branches &
  \path{alignAndRound?_refines} \\
\midrule
\rowcolor{black!6}
\multicolumn{2}{@{}l@{}}{\strut\textit{Quotient validation and rounding}\quad
  \path{FloatLib.Numerics.FixedWord}} \\
Check the quotient and remainder &
  \path{CertifiedDivision.certificate_sound}\newline
  \path{CertifiedDivision.certificate_complete} \\
Prove the fallback and its capacity &
  \path{RestoringQuotient.quotientSteps128_spec}\newline
  \path{CertifiedDivision.checkedCandidate_sound} \\
Round the exact quotient &
  \path{CertifiedDivision.roundQuotient_toNat} \\
\midrule
\rowcolor{black!6}
\multicolumn{2}{@{}l@{}}{\strut\textit{Encoded operations and real interpretation}\quad
  \path{FloatLib.Floats}} \\
Certify complete subtraction &
  \path{Formats.BinaryInterchange.Configured.Plan.wideLimbSub} \\
Preserve the specification through selection &
  \path{ExecFloat.Backend.selectCertified_run_eq_spec} \\
Apply the real FMA theorem &
  \path{Formats.BinaryInterchange.Model.toReal_fma_eq_roundAt} \\
\bottomrule
\end{tabularx}
\end{table}

\subsection{Jamming, normalization, and packing}

To prove the suffix invariant in \Cref{thm:sticky}, decompose, for $k>0$,
\[
 x=2^j(A2^k+B)+v,\qquad 0\le B<2^k,\quad 0\le v<2^j.
\]
If $v=0$, jamming returns $A2^k+B$. Otherwise it returns
$A2^k+(B\mathbin{\mathrm{OR}}1)$. Since $k>0$, setting bit zero leaves
the suffix below $2^k$. In either case the quotient is $A$, and the
suffix is zero exactly when $B=v=0$. This proves the
quotient and suffix identities together; taking the quotient modulo two
also proves the positive-index bit identity. At the guard position
$k=s-j-1$, the condition $j+2\le s$ makes this argument applicable.
It is a uniform separation condition: for every $j\ge1$,
$x=2^{j+1}+1$ gives $J_j(x)=3$ and
$E_1(3)=2\ne E_{j+1}(x)=1$. Particular inputs can still agree without
separation; in particular, $J_0$ is the identity.

An operation must also preserve the exponent used for packing.
For $x\ge2^{j+1}$, the truncated quotient is at least two, and
$\lfloor\log_2 J_j(x)\rfloor=\lfloor\log_2x\rfloor-j$.
Indeed, low-bit insertion either leaves the quotient unchanged or
increments an even quotient. The latter cannot cross a power-of-two
boundary, since the integer immediately below that boundary is odd.

Write $L=\lfloor\log_2x\rfloor$. The stronger bound
$x\ge2^{p+j+2}$ ensures $L-j\ge p+2$, so the jammed rounder's shift is
$(L-j)-p=(L-p)-j$. Its rounded significand therefore equals
$m=E_{L-p}(x)$ by \Cref{thm:sticky}.
Write $w$ for exponent width and $h=2^{w-1}-1$ for conventional bias.
For an IEEE-style descriptor eligible for the wide-limb backend,
$1+w+p>128$ and $w\le32$. The unsigned argument \lean{scale},
denoted $\sigma$, represents the unrounded magnitude
$x\,2^{\sigma-2(h+p-1)}$, with exponent subtraction in $\Z$.
The jammed rounder retains this argument and adds $j$ when computing
the packing position: $(L-j)+j+\sigma=L+\sigma$.
Write $\rho=L+\sigma$ for this common position.
Their common significand satisfies $2^p\le m\le2^{p+1}$, so the
normalization carry is exactly $c=\mathbf{1}[m=2^{p+1}]$.
The normal rounder accepts exactly when
\begin{equation}\label{eq:appendix-normal-range}
 h+2p-1\le \rho,\qquad \rho+c\le3h+2p-2.
\end{equation}
The packed fields are $e=\rho+c-(h+2p-2)$ and $f=m\bmod2^p$.
The acceptance bounds give $1\le e\le2h=2^w-2$ and $0\le f<2^p$.
Since $w\le32$, conversion of $e$ to \lean{UInt32} is exact,
and taking the low $p$ limb bits computes $f$ exactly.
If $c=0$, then $f=m-2^p$. If $c=1$, then $m=2^{p+1}$ and $f=0$:
the increased exponent accounts for renormalization.
These bounds and field equations identify the packed result.
The natural-number theorem separately handles
$j=0$ with only $x\ne0$, including left normalization when $L<p$.
The limb refinement then gives the same model result for an array
denoting $J_j(x)$ whenever the normal branch accepts. Declined results
are handled by the complete operation.

\subsection{Limb subtraction and cancellation}

Let $a,b\ge2^p$, $s_b\le s_a$, $d=s_a-s_b$, and
$\ell_a=\lfloor\log_2a\rfloor$, $\ell_b=\lfloor\log_2b\rfloor$.
In the opposite-sign compressed branch, $d\ge3$ and
$\ell_b+2\le\ell_a+d$. Put $j=d-3$, $T=a2^d-b$, and
$b=Q2^j+r$ with $0\le r<2^j$. The guard supplies two distinct bounds:
\begin{equation}\label{eq:appendix-borrow-bounds}
 \begin{aligned}
 T&\ge2^{\ell_a+d-1}\ge2^{p+j+2},\\
 Q&<2^{\ell_a+2}\le4a,\qquad Q+1\le8a.
 \end{aligned}
\end{equation}
The first supplies the rounding separation. The second makes the
subtraction of $Q$ and the incoming borrow $\mathbf{1}[r\ne0]$
nonnegative. It follows by dividing
$b<2^{\ell_a+d-1}=2^{\ell_a+2}2^j$.

To justify that subtraction for any limb count, set $\beta=2^{32}$.
Zero-extend the shorter array. For digits $u_i,v_i<\beta$ and incoming
borrow $b_i\in\{0,1\}$, the loop computes
\[
 \begin{aligned}
 t_i&=u_i+\beta-v_i-b_i,\quad
 z_i=t_i\bmod\beta,\quad b_{i+1}=\mathbf{1}[t_i<\beta],\\
 u_i+\beta b_{i+1}&=z_i+v_i+b_i.
 \end{aligned}
\]
Here $0\le t_i<2\beta<2^{64}$, so the 64-bit intermediate is exact.
The induction preserves nonnegativity of the unprocessed suffix:
if $V+b_i\le U$ and $U=u_i+\beta U'$, $V=v_i+\beta V'$, the
displayed identity and $z_i<\beta$ imply $V'+b_{i+1}\le U'$.
After the last digit this forces the final borrow to vanish.
Consequently, for arrays denoting $U,V$ and borrow $b_0\le1$,
the result denotes $U-V-b_0$ whenever $V+b_0\le U$.
The allocated length is the maximum of the two input lengths; neither
equal lengths nor a fixed upper bound on $U,V$ is required.

Apply this contract with $U=8a$, $V=Q$, and $b_0=\mathbf{1}[r\ne0]$.
The complemented suffix $2^j-r$ in the inexact case, as derived in
\Cref{sec:math-sticky}, explains why the borrow precedes low-bit
insertion. The resulting limb value is exactly $J_j(T)$.
The positive minuend ensures a nonempty result array, which is the
storage premise for inserting the low bit.

When the dominance guard fails, alignment forms $a2^d$ exactly and
compares it with $b$ before subtracting. This retains arbitrarily deep
cancellation, including exact equality, which returns positive zero for
opposite signs. For an eligible descriptor, the accepted-branch theorem
permits arbitrary scales and signs and imposes only $a,b\ge2^p$ on the
magnitudes. Thus an exact FMA product $m_xm_y\ge2^{2p}\ge2^p$ meets the
same contract even though it can exceed the stored significand width.
Its limb multiplication and alignment introduce no intermediate rounding.

\subsection{Quotient certification and repair}

Put $M=2^{128}$. The checker takes $n,d,q,r<M$, with $64<p\le126$
and $s\in\{p,p+1\}$. Its soundness and completeness theorems give
acceptance exactly when $qd+r=n2^s$ and $r<d$.
They require neither normalized inputs nor the comparison-directed
choice of $s$. The widening lemma uses $64<s<128$; in particular,
$n2^s<2^{128+s}\le2^{255}<M^2$.

If $u,c$ denote the value and carry of the four-word sum, multiplication
and addition refinement give $u+cM^2=qd+r$. The soundness proof uses
the checked conjunct $c=0$ to recover exact addition. Nevertheless,
the argument widths alone imply
\begin{equation}\label{eq:appendix-checker-capacity}
 qd+r\le(M-1)^2+(M-1)=M^2-M<M^2.
\end{equation}
The carry test is therefore redundant for these types, including
rejected candidates; it does not exclude any representable Euclidean
decomposition. The remainder bound implies $d>0$, and uniqueness of
Euclidean division identifies $q$ and $r$.

Repair needs the stronger domain $2^p\le n,d<2^{p+1}$, with $s=p$
when $d\le n$ and $s=p+1$ otherwise. It starts from
$q_0=\mathbf{1}[d\le n]$, $r_0=n-q_0d$. Normalization gives $n<2d$,
so $r_0<d$. After $i$ steps, the loop invariant is
\begin{equation}\label{eq:appendix-restoring-invariant}
 q_i d+r_i=n2^i,\qquad r_i<d,\qquad
 (q_i+1)2^{s-i}\le M\quad(0\le i\le s).
\end{equation}
Initially $(q_0+1)2^s=2^{p+1}\le M$.
A step sets $b=\mathbf{1}[2r_i\ge d]$,
$q_{i+1}=2q_i+b$, and $r_{i+1}=2r_i-bd$.
Because $d<2^{127}$, doubling the remainder fits in 128 bits.
If a step remains, the capacity invariant implies $2q_i+1<M$;
the inequality $q_{i+1}+1\le2(q_i+1)$ then preserves capacity for
the remaining steps. At termination the pair satisfies the checker.
The runtime selects this proved repair directly after rejection,
without executing a second check.

The quotient-rounding theorem requires the strict increment bound
$q+1<M$, together with $qd+r=N$, $r<d$, and $d<2^{127}$,
for arbitrary natural $N$.
These bounds justify exact doubling of $r$ and incrementing $q$.
The normalized caller supplies them because its scaled quotient lies
in $[2^p,2^{p+1})$, giving $q+1\le2^{p+1}\le2^{127}<M$.
The comparisons of $2r$ with $d$, including quotient parity at equality,
then compute exact nearest-even rounding of $N/d$.

\para{Encoded refinement.}
For an eligible implementation, let $D,C$ be the codec's decoding and
encoding maps, $S$ the model specification, and $K$ the complete binary
operation. Combining the accepted-branch proof with the reference branch gives
$D(K(x,y))=S(Dx,Dy)$; hence
$K(x,y)=C(D(K(x,y)))=C(S(Dx,Dy))$.
This is equality of encoded results, including the specified zero signs
and exceptional encodings. The reference baseline executes its supplied
specification, so its refinement certificate is reflexive; selection
preserves that same specification. Numerical correctness comes from
the separate model theorem. In particular, the real FMA corollary uses
an IEEE-style descriptor, finite operands, and a finite result, whereas
the complete encoded equality has no operand restrictions.

\subsection{Benchmark fixtures and wrapper exclusion}
\label{app:evaluation-details}

The scalar experiment in \Cref{sec:evaluation} prepares sixteen tuples.
For $0\leq i<16$, define
\begin{equation}\label{eq:benchmark-fixtures}
 r_s(i)=\frac{((is+s+1)\bmod113)+7}{((11i+s)\bmod29)+32}.
\end{equation}
The first operand uses $r_{37}(i)$, negated when $i\bmod5=0$;
the second uses $r_{61}(i)$, negated when $i\bmod3=0$.
Square root uses positive $r_{43}(i)$; FMA's third operand reverses the
first sequence. All inputs are rounded before timing.

The excluded Flocq~4.2.2 wrapper extracts \texttt{BinarySingleNaN}
arithmetic to OCaml/Zarith \cite{flocq422}.
It matches significand precision but uses exponent bound 16384 and
rounds fixture numerators and denominators separately. The first operand
is $-r_{37}(0)=-45/40$. At significand precision two, nearest-even rounding
gives $R_2(-45)=-48$, $R_2(40)=32$, and
\begin{equation}\label{eq:flocq-fixture-mismatch}
 R_2\bigl(R_2(-45)/R_2(40)\bigr)=-1.5
 \ne R_2(-45/40)=-1.
\end{equation}
This mismatch can change the subsequent dependent input sequence.
\section{Exact comparison and quire invariants}
\label{app:rounding}

\Cref{thm:posit-comparison} bounds the finite code search; comparing its
real target with a rational boundary needs a separate termination argument.

\subsection{Termination of enclosure comparison}

For intervals $I_k=[L_k,U_k]$ with rational endpoints and a boundary
$t\in\Q$, the comparison procedure finds
$k_*=\min\{k\in\N:U_k<t\ \lor\ t<L_k\}$ and returns $<$ or $>$,
respectively. The \lean{Nat.find} search takes an existence proof,
erased during compilation.

\begin{lemma}[Soundness and termination of enclosure comparison]
\label{lem:enclosure-comparison}
Suppose $L_k\leq y\leq U_k$ for every $k$, both endpoint sequences
converge to $y$, and $y\neq t$. Then $k_*$ exists and the returned
ordering is $\operatorname{cmp}(y,t)$.
\end{lemma}
\begin{proof}
If $y<t$, set $\varepsilon=(t-y)/2>0$. Convergence of the upper
endpoints eventually gives $|U_k-y|<\varepsilon$, hence
$U_k<y+\varepsilon<t$. If $t<y$, convergence of the lower endpoints
similarly gives $t<L_k$. Thus one of the tests eventually succeeds.
At the first successful test, containment proves the reported order:
$y\leq U_{k_*}<t$ or $t<L_{k_*}\leq y$.
\end{proof}

Neither nested intervals nor a known convergence rate is required.
Equality needs a separate decision: the intervals
$[1-2^{-k},1+2^{-k}]$ converge to the boundary $t=1$ but never
separate it. Each interval admits targets on both sides of $t$, so a
fixed refinement limit cannot decide equality.

For rational $q\neq0$, irrationality of $\exp q$ excludes equality with
rational boundaries; for $\log q$ we need $q>0$ and $q\neq1$. Together
with endpoint convergence at degrees $2^k$ and direct comparison of
$\exp0=1$ and $\log1=0$, this yields total comparators.
Within \path{FloatLib.Numerics.Enclosure},
\path{Comparison.exists_separating} proves termination and
\path{Comparison.compare_eq_real} proves soundness.
The lemmas \path{exists_exp_separating} and \path{exists_log_separating}
specialize termination to the exponential and logarithm.
The resulting comparison equation supplies the hypothesis of
\Cref{thm:posit-comparison}.

\subsection{Enclosure reuse}

Rounding reuses Taylor enclosures across boundary queries for the same
target. We prepare a finite array of lazy enclosures outside the boundary
function, so the queries share evaluated entries.
For every cache length $c$, \lean{cacheIntervals\_apply} proves
$\operatorname{get}(\operatorname{cache}(I,c),k)=I_k$ for all $k$:
entries with $k\geq c$ call the original generator, so the cache length
does not limit refinement.

For a positive rational threshold $t$,
$\exp q<t\Longleftrightarrow q<\log t$.
We reverse an exact comparison of $\log t$ with $q$ to avoid constructing
a potentially enormous exponential; nonpositive thresholds are decided
immediately because $\exp q>0$. The prepared comparator uses direct
exponential enclosures for moderate arguments and the logarithmic route
otherwise. Both satisfy the same comparison equation.
The scalar timing experiment excludes these elementary-function
implementations.

\subsection{Quire accumulation from a partial sum}

For a standard width $n\geq2$, let an ordinary quire coefficient $A$
denote $v=A\Delta_n$, with $\Delta_n=2^{16-8n}$.
Suppose $m,k\in\N$, $|A|\leq mB_n$, and $m+k<2^{31}$,
where $B_n=2^{16n-32}$ and $m$ counts the product bounds already consumed.
For $k$ remaining ordinary operand pairs with exact products
$p_0,\ldots,p_{k-1}$, the fold continuation represents
$v+\sum_{i<k}p_i$. After $j\leq k$ updates, the induction invariant gives
the exact value $v+\sum_{i<j}p_i$ and a coefficient $A_j$ satisfying
\[
 |A_j|\leq(m+j)B_n<2^{31}B_n=2^{16n-1}.
\]
The initial state satisfies this invariant by assumption.
The triangle inequality advances the first bound, and
$m+j\leq m+k<2^{31}$ gives the second. Every state is ordinary, so the
one-step theorem preserves exactness. The model theorem
\lean{toRat?\_foldMulAdd} formalizes this induction;
$A=v=m=0$ recovers \Cref{thm:quire}.

Ordinary coefficients satisfy $-2^{16n-1}<A<2^{16n-1}$ because the
most-negative two's-complement word is NaR. Maximum-magnitude products
attain $B_n$ quire units, so $2^{31}$ same-sign maximum products reach a
boundary: the positive coefficient overflows, and the negative
coefficient is the reserved NaR word. The update returns NaR in both cases.
The strict bound covers every operand ordering, although cancellation can
allow longer sums.

At the configured interface,
\path{ExecFloat.Posit.Quire.toRat?_foldl_qMulAdd_zero} transfers the fold
theorem through the operand codec. Its successful rational decoding
identifies the exact product sum and excludes NaR.
The model theorem \lean{qToP\_eq\_roundRat} and codec round trip give one
rounding of the sum;
\lean{RealRounding.round\_ratCast} identifies it with $\Pi_n$.

\end{document}